\documentclass[journal,onecolumn,draftcls,12pt]{IEEEtran}

\usepackage{color}
\definecolor{gray}{rgb}{0.5,0.5,0.5}

\usepackage{amsmath,amsthm}
\usepackage{amssymb,amsfonts}
\usepackage{graphicx}
\usepackage{url}
\usepackage{subfigure}
\usepackage{bm}
\usepackage{subeqnarray} 
\usepackage{relsize}  
\usepackage[ruled,lined,linesnumbered]{algorithm2e}
\usepackage{paralist} 
\usepackage{hyperref} 

\newtheorem{theorem}{Theorem}

\makeatletter
\def\IEEElabelanchoreqn#1{\bgroup
	\def\@currentlabel{\p@equation\theequation}\relax
	\def\@currentHref{\@IEEEtheHrefequation}\label{#1}\relax
	\Hy@raisedlink{\hyper@anchorstart{\@currentHref}}\relax
	\Hy@raisedlink{\hyper@anchorend}\egroup}
\makeatother
\newcommand{\subnumberinglabel}[1]{\IEEEyesnumber
	\IEEEyessubnumber*\IEEElabelanchoreqn{#1}}

\usepackage{cite}
\ifCLASSINFOpdf

\else

\fi

\begin{document}

\title{Interference Cancellation Information Geometry Approach for Massive MIMO Channel Estimation}

\author{An-An Lu, \IEEEmembership{Member, IEEE}, Bingyan Liu and Xiqi Gao, \IEEEmembership{Fellow, IEEE}

\thanks{A.-A. Lu, B. Liu and X. Q. Gao are with the National Mobile Communications Research Laboratory (NCRL), Southeast University,
Nanjing, 210096 China, and also with Purple Mountain Laboratories, Nanjing 211111, China, e-mail: aalu@seu.edu.cn, xqgao@seu.edu.cn.} 
}


\maketitle

\begin{abstract}
	In this paper, the interference cancellation information geometry approaches (IC-IGAs) for massive MIMO channel estimation are proposed.
	The proposed algorithms are low-complexity approximations of the minimum mean square error (MMSE) estimation.
	To illustrate the proposed algorithms, a unified framework of the information geometry approach for channel estimation and its geometric explanation are described first.
	Then, a modified form that has the same mean as the MMSE estimation is constructed. Based on this, the IC-IGA algorithm 
	and the interference cancellation simplified information geometry approach (IC-SIGA) are derived by applying the information geometry framework.
The \textit{a posteriori} means on the equilibrium of the proposed algorithms are proved to be equal to the mean of MMSE estimation, 
	and the complexity of the IC-SIGA algorithm in practical massive MIMO systems is further reduced by considering the beam-based statistical channel model (BSCM) and fast Fourier transform (FFT).
	Simulation results show that the proposed methods achieve similar performance as the existing information geometry approach (IGA) with lower complexity.

\end{abstract} 

%
\IEEEpeerreviewmaketitle

\newpage
\section{Introduction}
Massive multi-input multi-output (MIMO)\cite{lu2014overview,bjornson2014massive, marzetta2016fundamentals} is the core enabling technology for the 5th generation (5G) mobile communications. It has further evolved into extra large-scale MIMO (XL-MIMO) \cite{de2020non, marinello2020antenna, wang2024tutorial}, which has become a research hotspot of the 6th generation (6G) mobile communications.
By increasing the number of antennas at the base station (BS), massive MIMO has significantly enhanced the spatial multiplexing and diversity gain and achieved a substantial increase in energy and spectral efficiency.
 To achieve these potential gains, the most important thing is acquiring accurate channel state information (CSI). In this paper, we focus on the channel estimation problem for massive MIMO.

The object of channel estimation is obtaining the \textit{a posteriori} information of the channel from received signals. When the \textit{a priori} probability density function (PDF) is Gaussian, the minimum mean square error (MMSE) estimator, which is also the \textit{a posteriori}  mean, is the optimal estimator. 
To fully exploit the sparsity of massive MIMO, analytical channel models with joint space-frequency representation such as the doubly beam-based stochastic model (BSCM)\cite{lu20242d} are
established, and the channel estimation problem can be transformed into the angle-delay domain. As the number of antennas increases, the pilot resources in massive MIMO are no longer enough \cite{elijah2015comprehensive, you2015pilot} and non-orthogonal pilots  \cite{wang2014design, yang2022channel} are often used. 
Thus, the channel estimator in massive MIMO usually needs to perform joint estimation of the channels of different users in the angle-delay domain, which makes the complexity of the matrix inversion in the MMSE estimation prohibitive. 

Due to the complexity issue, low-complexity channel estimators that can achieve near MMSE performance are widely investigated in the literature. In \cite{shariati2014low}, a polynomial expansion (PE) channel estimation is proposed for massive MIMO with arbitrary statistics. In \cite{wang2019channel}, a low-complexity channel estimation with low dimensional channel gain estimation is proposed for massive MIMO with uniform planar array (UPA) by estimating the angle of arrival (AOA) first. 
Deep learning-based channel estimation approaches are proposed for beamspace mmWave massive MIMO and multi-cell massive MIMO systems in \cite{he2018deep} and \cite{balevi2020massive}, respectively.
In \cite{bellili2019generalized}, a generalized approximate message passing (GAMP) method is proposed for channel estimation of a massive MIMO mmWave channel. 
Among these approaches, the GAMP might be the most promising one to be implemented in practical systems. However, the derivation of the GAMP is not easy to follow since it lacks of rigorous and concise mathematical explanation.
 In \cite{yang2022channel}, an information geometry approach (IGA) that can achieve similar performance with similar complexity as the GAMP is proposed for massive MIMO channel estimation.

Information geometry theory arises from the study of invariant geometric structures in statistical inference \cite{amari2016information} and provides the mathematical foundation of statistics \cite{ay2017information}. It views the space of probability distributions as manifold and tackles problems in information science by using the concepts of differential geometry with tensor calculus \cite{nielsen2022many}.   
	Information geometry theory also plays an important role in machine learning, signal processing, optimization theory\cite{amari2010information}, and fields such as neuroscience \cite{oizumi2016unified} and quantum physics \cite{banchi2014quantum}.
The information geometry explanation of the belief propagation (BP) algorithm \cite{pearl1988probabilistic} is given in \cite{ikeda2004stochastic}, which also shows that the concave-convex procedure (CCCP) method \cite{yuille2002cccp} computing the marginal distribution can be interpreted by information geometry.
In \cite{ikeda2004information}, the decoding algorithms for Turbo codes and low-density parity check (LDPC) codes are derived from the viewpoint of information geometry, and the equilibrium and error of the algorithms are analyzed.

The research on MIMO and massive MIMO based on information geometry is rarely seen in the literature. In \cite{zia2007information}, an information geometric approach is proposed to approximate ML estimation in semi-blind MIMO channel identification.
For massive MIMO, information geometry is introduced in \cite{yang2022channel} and  \cite{yang2024signal} to derive information geometry approaches (IGA) for channel estimation and detection, respectively.
Moreover, a simplified IGA (S-IGA) algorithm is proposed in \cite{yang2023channel} by using the constant envelope property of the channel measurement matrix. 
Information geometry provides a unified framework for understanding belief propagation or message passing based algorithms. 
The detailed relation between the IGA and AMP algorithm is provided in [In preparation].

The information geometry approach define the auxiliary probability density functions (PDFs) based on the orignal PDF to obtain a low complexity algorithm.
 In \cite{yang2022channel}, The auxiliary PDFs are defined based on the elements of the received signal, and each auxiliary PDF computes the message of all the channel elements. 
 Thus, a natural question is whether we can derive a new channel estimation algorithm that is different from and has a lower complexity than that in \cite{yang2022channel}. To answer this question, we propose the interference cancellation information geometry approach (IC-IGA) for massive MIMO channel estimation in this paper. 
 In the new algorithm, each auxiliary PDF focuses on the message for one element of the channel vector, and both time and space complexities are much lower than that of the IGA algorithm.  To derive this new algorithm, we first provide a unified framework of the channel estimation information geometry approach, and explain the geometric meaning of the equilibrium of this approach. Then, we construct a modified channel estimation form that has the same mean as the MMSE estimation and apply the unified framework to obtain the new IG algorithm. To further reduce the complexity, the interference cancellation simplified information geometry approach (IC-SIGA) is proposed.
Finally, the \textit{a posteriori} means on the equilibriums of the proposed algorithms are proved to be equal to the mean of MMSE estimation, and the complexity analysis is provided.

The rest of this paper is organized as follows. The preliminaries about the manifold of complex Gaussian distributions are provided in Section II. 
The general information geometry framework for massive MIMO channel estimation is presented in
Section III. 
The derivations of IC-IGA and IC-SIGA are presented in
Sections IV and V, respectively. 
Simulation results are provided in Section VI. 
The conclusion is drawn in Section
VII.

{\it Notations}: Throughout this paper, uppercase and lowercase boldface letters are used for matrices and vectors, respectively. The superscripts $(\cdot)^*$, $(\cdot)^T$, and $(\cdot)^H$ denote the conjugate, transpose, and conjugate transpose operations, respectively. The mathematical expectation operator is denoted by ${\mathbb E}\{\cdot\}$. 
The operators $\det(\cdot)$ represent the matrix determinant, and $\|\cdot\|_2$ is the $\ell_2$ norm. The operators $\odot$ and $\otimes$ denote the Hadamard and Kronecker product, respectively. The $N \times N$ identity matrix is denoted by $\mathbf I_N$, and
 $\mathbf I_{N, M}$ is used to denote $[\mathbf I_N ~ \mathbf 0_{N,(M-N)} ]$ when $N<M$ and $[\mathbf I_M ~ \mathbf 0_{M,(N-M)} ]^T$ when $N>M$. 
A vector composed of the diagonal elements of $\mathbf X$ is denoted by $\text{diag}(\mathbf X)$, and a diagonal matrix with $\mathbf x$ along its diagonal is denoted by $\text{diag}(\mathbf x)$.
We use $h_n$ or $[\mathbf h]_n$, $a_{mn}$ or $[\mathbf{A}]_{mn}$, $[\mathbf{A}]_{:,n}$ and $[\mathbf{A}]_{m,:}$ to denote the $n$-th element of the vector $\mathbf h$, the $(m,n)$-th element of the matrix $\mathbf{A}$, the $n$-th column and the $m$-th row of matrix $\mathbf{A}$, respectively.
The symbol $\lceil x \rceil$ denotes the smallest
integer among those larger than $x$.
Define $\mathbb Z_{N}^+ = \{ 0, 1, \cdots, N\}$. 
The operation $a\bmod b$ denotes the integer $a$ modulo the integer $b$.

\section{Preliminaries}

In this section, we present an information geometric perspective on the space of multivariate complex Gaussian distributions by using the concepts from \cite{amari2016information}.

\subsection{Affine and Dual Affine Coordinate Systems}

From information geometry theory, we have that the manifold of complex Gaussian distribution is a dually flat manifold, which has an affine coordinate system and a dual affine coordinate system.
The two coordinate systems are also called natural parameters and expectation parameters in the literature.  In the following, we describe the two affine coordinate systems in detail.

The Gaussian distributions belong to the exponential family of distributions \cite{simeone2022machine}. Let $\bm\theta, \bm\Theta$ be the natural parameter of a complex Gaussian distribution of a random vector $\mathbf x \in \mathbb C^{N \times 1}$, then the PDF $p(\mathbf{x}; \bm\theta, \bm\Theta)$ is defined as
\begin{IEEEeqnarray}{Cl}
 	\label{eq:gaussComp}
 	p( \mathbf x;\bm\theta, \bm\Theta) 
 	= \exp\left\{  \mathbf x^H \bm\theta + \bm\theta^H \mathbf x + \mathbf x^H\bm\Theta\mathbf x -  \psi(\bm\theta, \bm\Theta)  \right\}
 \end{IEEEeqnarray}
where $\psi(\bm\theta, \bm\Theta)$ is the normalization factor, which is called free energy function and given by
 \begin{IEEEeqnarray}{Cl}
	\label{eq:free_energy_function}
	\psi(\bm\theta, \bm\Theta) 
	&= N\log(\pi) - \log\det(-\bm\Theta) - \bm\theta^H\bm\Theta^{-1}\bm\theta.
\end{IEEEeqnarray}
Let $\mathcal M = \{p( \mathbf x;\bm\theta, \bm\Theta) \}$ be the manifold of multivariate complex Gaussian distributions and  $\bm\theta, \bm\Theta$ is an coordinate system.  
The free energy function $\psi(\bm\theta, \bm\Theta)$ is a convex function of $\bm\theta, \bm\Theta$ and introduces an affine flat structure, which means the $\bm\theta, \bm\Theta$ is an affine coordinate system and each coordinate axis
of $\bm\theta, \bm\Theta$ is a straight line. Furthermore, the Bregman divergence \cite{bregman1967relaxation} from $\bm\theta, \bm\Theta$ to $\bm\theta', \bm\Theta'$ derived from $\psi(\bm\theta, \bm\Theta)$ is the same as the Kullback–Leibler (KL) divergence from $p( \mathbf x;\bm\theta', \bm\Theta')$ to $p( \mathbf x;\bm\theta, \bm\Theta)$.

 
The dual coordinate system of $\mathcal M$ is obtained by the Legendre transformation \cite{amari2000methods}, \textit{i.e.}, the gradients of the free energy function $\psi(\bm\theta, \bm\Theta)$.  
From \eqref{eq:free_energy_function}, we can obtain the dual coordinate $\bm\mu,\mathbf M$ as
\begin{IEEEeqnarray}{Cl}
	\subnumberinglabel{eq:dualmuM}
	\frac{\partial \psi }{\partial \bm\theta^*} 
	&= \mathbb E\{ \mathbf x\}
	= \bm\mu
	= -\bm\Theta^{-1}\bm\theta
	\label{eq:dual_mu}\\
	\frac{\partial \psi }{\partial \bm\Theta} 
	&= \mathbb E\{ \mathbf x \mathbf x^H\}
	= \mathbf M
	= \bm\Theta^{-1}\bm\theta \bm\theta^H\bm\Theta^{-1}
	- \bm\Theta^{-1}
	= \bm\mu\bm\mu^H + \bm\Sigma
	\label{eq:dual_M} 
\end{IEEEeqnarray}
where $\bm\Sigma$ is the covariance matrix of $\mathbf x$.
The dual coordinate is the combination of the first and second order moments of $\mathbf x$ and is also called the expectation parameter.
The dual function of $\psi$ is given by 
\begin{IEEEeqnarray}{Cl} 
	\phi=\psi^* 
	&= \int p(\mathbf x;\bm\theta,\bm\Theta) 
	\log {p(\mathbf x;\bm\theta,\bm\Theta)} \, \mathrm d \mathbf x
\end{IEEEeqnarray}
which is the negative entropy of the PDF $p(\mathbf x;\bm\theta,\bm\Theta)$.
By using the dual coordinate system, we have that
\begin{IEEEeqnarray}{Cl}
  \phi(\bm\mu, \mathbf M) &=  c - \log\det(\mathbf{M} - \bm\mu\bm\mu^H)   
\end{IEEEeqnarray}
where $c$ is a constant.
The dual function $\phi$ is a convex function of $\bm\mu, \mathbf M$ and induces the dual affine flat structure. The Bregman divergence  from $\bm\mu, \mathbf M$ to $\bm\mu', \mathbf M'$ derived from $\phi$ is the KL divergence from $p(\mathbf{x};\bm\mu, \mathbf M)$ to $p(\mathbf{x};\bm\mu', \mathbf M')$. 

 The transformations from expectation parameters to natural parameters can also be obtained from the Legendre transformation as
\begin{IEEEeqnarray}{ClCl} 
\subnumberinglabel{eq:transform_nat_from_exp}
	\bm\theta &= \frac{\partial \phi }{\partial \bm\mu^*} = \bm\Sigma^{-1}\bm\mu  \qquad \\
	 \bm\Theta &= \frac{\partial \phi }{\partial \mathbf{M}} = -\bm\Sigma^{-1} 
\end{IEEEeqnarray}
where $\bm\Sigma=\mathbf{M} -\bm\mu\bm\mu^H$ is the covariance matrix and is used here for brevity.
By using the expectation parameter, the PDF can also be written in the familiar form as
\begin{IEEEeqnarray}{Cl}
	\label{eq:gaussCompEx}
	p(\mathbf x;\bm\mu,\bm\Sigma) 
	&= \exp\left\{ \mathbf x^H\bm\Sigma^{-1}\bm\mu
	+  \bm\mu^H\bm\Sigma^{-1}\mathbf x
	- \mathbf x^H\bm\Sigma^{-1}\mathbf x
	-  \psi(\bm\mu,\bm\Sigma) \right\}
\end{IEEEeqnarray}
where $\psi(\bm\mu,\bm\Sigma) = \log(\pi^N) + \log\det(\bm\Sigma) + \bm\mu^H\bm\Sigma^{-1}\bm\mu$.

\subsection{$e$-flat Submanifold and $m$-Projection}
\label{eflat_and_mproj}

After introducing the affine and dual affine coordinate systems, we present the definitions of $e$-flat and $m$-projection, which are very important when describing the information geometry approach for channel estimation.

A submanifold $\mathcal M_1 \subset \mathcal M$ is called $e$-flat if it has a linear constraint in the affine coordinate $\bm\theta, \bm\Theta$. The term $m$-flat can be similarly defined by using the dual affine coordinate $\bm\mu, \mathbf{M}$.
An example of  $e$-flat submanifold is the manifold of independent complex Gaussian distributions, defined as
\begin{IEEEeqnarray}{Cl}
	\mathcal M_0 
	= \{p( \mathbf x;\bm\theta_0, \bm\Theta_0)\} 
\end{IEEEeqnarray}
where  $\bm\Theta_0 \in \mathbb R^{N \times N}$ is a diagonal matrix. Since $\bm\Theta_0$ are diagonal matrices, the expectation parameter is very easy to obtain  as
\begin{IEEEeqnarray}{ClCl} 
	\bm\mu_0 &= -\bm\Theta^{-1}\bm\theta_0, \qquad
	&\bm\Sigma_0 &= -\bm\Theta_0^{-1}.  
\end{IEEEeqnarray}

The projection to an $e$-flat manifold is called $m$-projection because the projection can be realized linearly in the dual affine coordinate system.
Let $p( \mathbf x;\bm\theta_0, \bm\Theta_0)$ and $p( \mathbf x;\bm\theta_1, \bm\Theta_1)$ be two points in the manifold $\mathcal M$, refered as $P_0$ and $P_1$, and $P_0 \in \mathcal M_0$  and $P_1 \notin \mathcal M_0$.  

The $m$-projection is unique and minimizes the KL divergence. 
Specifically, the $m$-projection from $P_1$ to $\mathcal M_0$ is defined as
\begin{IEEEeqnarray}{Cl}
	p(\mathbf x;\bm\theta_1^{0},\bm\Theta_1^{0})
	&= 
	\mathlarger{\Pi}_{\mathcal M_0}^m \{p(\mathbf x;\bm\theta_1,\bm\Theta_1)\} \notag\\
	&= \mathop{\arg\min}\limits_{p(\mathbf x;\bm\theta_0,\bm\Theta_0) \in \mathcal M_0} D_{KL}\left(p(\mathbf x;\bm\theta_1,\bm\Theta_1) ; p(\mathbf x;\bm\theta_0,\bm\Theta_0)\right).
\end{IEEEeqnarray}

By using the affine and dual affine coordinate systems, the $m$-projection is easy to obtain. First, rewrite the KL divergence as \cite{amari2016information}
 \begin{IEEEeqnarray}{Cl}
& D_{KL}(P_{1} ; P_{0}) =  \phi(\bm\mu_1, \mathbf{M}_1) + \psi(\bm\theta_0, \bm\Theta_0) - \bm\mu_1^H\bm\theta_0   -  \bm\theta_0^H\bm\mu_1 -{\rm tr}(\mathbf{M}_1\bm\Theta_0).  
\label{eq:KL_Divergence}
\end{IEEEeqnarray} 
Then, the minimum can be obtained from the first-order optimal condition  
\begin{IEEEeqnarray}{Cl}
	\subnumberinglabel{eq:projection}
	\frac{\partial D_{KL}(P_{1} ; P_{0}) }{\partial \bm\theta_0^*} 
	&= \frac{\partial \psi }{\partial \bm\theta_0^*} 
	- \bm\mu_1 =  \bm\mu_0 -\bm\mu_1 \\
	\frac{\partial D_{KL}(P_{1} ; P_{0}) }{\partial \bm\Theta_0} 
	&= \frac{\partial \psi }{\partial \bm\Theta_0} 
	 - \frac{\partial {\rm tr}(\mathbf{M}_1\bm\Theta_0)  }{\partial \bm\Theta_0}  =   \mathbf I\odot {\mathbf M}_0  
	- \mathbf I\odot {\mathbf M}_1
\end{IEEEeqnarray}
where $\mathbf I\odot {\mathbf M}_1$ is obtained because $\bm\Theta_0$ is a diagnal matrix. Thus, we have $\bm\mu_0=\bm\mu_1$ and $\mathbf I\odot {\mathbf M}_0  
	= \mathbf I\odot {\mathbf M}_1$ for the projection point, and $\bm\theta_1^{0},\bm\Theta_1^{0}$ are their dual coordiantes. The projection is simple in the dual affine coordinate system. Let $P_1^0$ be the projection point, the dual straight line connecting $P_1$ and $P_1^0$  is the shortest one among those dual straight lines from  $P_1$ to $M_0$.

\section{Information Geometry Framework for Massive MIMO Channel Estimation}
In this section, we provide a framework of information geometry methods for channel estimation in massive MIMO systems inspired by Section 11.3.3 and 11.3.4 in \cite{amari2016information}.

\label{sec:IGAproce}

\subsection{Problem Formulation}
In massive MIMO channel estimation, a general received signal model is given by
\begin{IEEEeqnarray}{Cl}\label{reModel3}
	\mathbf y= \mathbf A\mathbf h+ \mathbf z
\end{IEEEeqnarray}
where 
$\mathbf A \in \mathbb C^{M \times N}$ is the deterministic measurement matrix, $\mathbf h$ is a random Gaussian vector distributed as $\mathbf h \sim \mathcal{CN}(\mathbf 0,\mathbf D)$, and $\mathbf z$ is the complex Gaussian noise vector distributed as $\mathbf h \sim \mathcal{CN}(\mathbf 0,\sigma_z^2\mathbf I)$.

For the received signal model in \eqref{reModel3},
the posterior distribution can be written as
\begin{IEEEeqnarray}{Cl}
	p(\mathbf h| \mathbf y)  
	&= \exp\left\{ 
	\sigma_z^{-2}\mathbf h^H\mathbf A^H\mathbf y 
	+ \sigma_z^{-2}\mathbf y^H \mathbf A\mathbf h -\mathbf h^H(\sigma_z^{-2}\mathbf A^H\mathbf A  + \mathbf{D}^{-1})\mathbf h -\psi \right\}.  \label{eq:Post_PDF}
\end{IEEEeqnarray}
According to \eqref{eq:gaussComp}, the natural parameters of $p(\mathbf h| \mathbf y)$ are
\begin{IEEEeqnarray}{Cl}
	\bm\theta &= \sigma_z^{-2} \mathbf A^H\mathbf y  
	\IEEEyesnumber\IEEEyessubnumber*\\
	\bm\Theta &=-(\sigma_z^{-2}\mathbf A^H\mathbf A  + \mathbf{D}^{-1})
\end{IEEEeqnarray}
whereas the dual affine coordinate can be obtained from \eqref{eq:dualmuM} as
\begin{IEEEeqnarray}{Cl}
	\bm\mu 
	&= -\bm\Theta^{-1}\bm\theta 
	=(\sigma_z^{-2}\mathbf A^H\mathbf A + \mathbf{D}^{-1})^{-1} \sigma_z^{-2} \mathbf A^H\mathbf y 
	\IEEEyesnumber\IEEEyessubnumber*\\
	\mathbf M 
	&= \bm\mu\bm\mu^H-\bm\Theta^{-1}
	= \bm\mu\bm\mu^H+(\sigma_z^{-2}\mathbf A^H\mathbf A + \mathbf{D}^{-1})^{-1}. 
\end{IEEEeqnarray}

The dual affine coordinate $\bm\mu$ is the posterior mean of $\mathbf h$, and thus is also the MMSE estimation. 
For massive MIMO systems, the complexity of $\bm\mu$ is often too high due to the inversion of the large dimensional matrix. 
Thus, one of the most important problems for massive MIMO is to derive low-complexity channel estimation methods.


\subsection{Information Geometry Framework}
From Section~\ref{eflat_and_mproj}, we know that if we can $m$-project $p(\mathbf{h};\bm\theta, \bm\Theta)$ onto the $e$-flat submanifold ${\cal{M}}_0$, then $\bm\mu$ is equal to the mean at the projection point, which is easy to obtain. However, the $m$-projection still involves the matrix inversion of the large dimensional matrix, and thus can not provide a low-complexity solution.

Information geometry provides other low-complexity ways to find a point in  ${\cal{M}}_0$ whose dual coordinate is approximation or the same as that of the $m$-projection point instead of using the $m$-projection. The IGA algorithm proposed in \cite{yang2022channel} is a specific algorithm derived based on information geometry, but its complexity can be further reduced. To extend the information geometry approach to derive new low-complexity algorithms, we provide a framework of information geometry for massive MIMO channel estimation. 

We call the submanifold ${\cal{M}}_0$ the target manifold since we want to find a target point in it. Since the target point can not be obtained directly by using the $m$-projection, the auxiliary manifolds and PDFs are needed to find a way to approximate the $m$-projection. 
The natural parameters or affine coordinates of original posterior PDF are $\bm\theta_{or} = \sigma_z^{-2}\mathbf{A}^H\mathbf y $,  $\bm\Theta_{or} = -(\sigma_z^{-2}\mathbf A^H\mathbf A  + \mathbf{D}^{-1})$.   
With the auxiliary manifolds, the process of the information geometry framework for channel estimation is summarized below:
\begin{compactenum}
	\item[\quad (1)]The natural parameters $\bm\theta_{or}$ and $\bm\Theta_{or}$ are split to construct $Q$ auxiliary manifolds of PDFs and one target manifold of PDFs;
	\item[\quad (2)]Initialize the auxiliary points and the target point in the auxiliary manifolds and target manifold, respectively;  
	\item[\quad (3)]Calculate the $m$-projections of the auxiliary points to the target manifold and compute the beliefs in the affine coordinate system;
	\item[\quad (4)]Update the natural parameters of the auxiliary and target points;
	\item[\quad (5)]Repeat (3) and (4) until the algorithm converges or fixed iterations, output the mean and variance of the target point.
\end{compactenum}

With the framework, the $m$-projection of the original point to the target manifold is approximated by the $m$-projections from the auxiliary points to the target manifold. To make the approximation well enough, two important conditions described in the following subsections are needed. 

\subsection{Split of Natural Parameter and the $e$-Condition}

In the general information geometry framework, the most important thing is to define the auxiliary points and manifolds. These definitions depend on the way of splitting natural parameter and determine what specific algorithm can be derived. 

The split of $\bm\theta_{or}$ and $\bm\Theta_{or}$ into $Q$ items is given as
\begin{IEEEeqnarray}{Cl}
	\subnumberinglabel{eq:IGsplit}
	\bm\theta_{or} = \sigma_z^{-2}\mathbf{A}^H\mathbf y 
	= \sum_{q=1}^Q \mathbf{b}_q \\
	\bm\Theta_{or}  = -(\sigma_z^{-2}\mathbf{A}^H\mathbf{A} + \mathbf{D}^{-1} )
	= -(\sum_{q=1}^Q \mathbf{C}_q + \bm\Lambda_c)
\end{IEEEeqnarray}
where the setting of $\mathbf{b}_q$, $\mathbf{C}_q$, $Q$, and $\bm\Lambda_c$ depends on specific algorithms, and $\bm\Lambda_c$ is usually a diagonal matrix. Let $\mathbf a_q = [\mathbf A^H]_{:,q}$ be the $q$-th column of $\mathbf A^H$. In the IGA algorithm proposed in \cite{yang2022channel}, the split $ \mathbf{b}_q = \mathbf a_q y_q$, $\mathbf{C}_q = \mathbf a_q \mathbf a_q^H$, $Q=M$, and $\bm\Lambda_c=\mathbf{D}^{-1}$ is used.

Based on the split, we define $Q$ auxiliary points or PDFs. Let $\bm\Lambda_q$ be a diagonal matrix.
The natural parameter of the $q$-th auxiliary point $p( \mathbf h; {\bm\theta}_q ,{\bm\Theta}_q) $ is defined by
\begin{IEEEeqnarray}{Cl}
	\subnumberinglabel{eq:tThetaq}
	{\bm\theta}_q &= {\bm\lambda}_q+\mathbf b_q
	\label{eq:thetaq}\\
	{\bm\Theta}_q &= - ({\bm\Lambda}_q + \mathbf C_q + \bm\Lambda_c).
	\label{eq:Thetaq}
\end{IEEEeqnarray}
where $\bm \lambda_q$ and $\bm\Lambda_q$ are variables in the natural parameter. 
The corresponding auxiliary PDF and manifold are then given as 
\begin{IEEEeqnarray}{Cl}\label{eq:auxpost}
	p( \mathbf h; {\bm\theta}_q , {\bm\Theta}_q)
	= 
	\exp\left\{ 
	\mathbf h^H(\bm\lambda_q + \mathbf b_q)
	+ (\bm\lambda_q + \mathbf b_q)^H\mathbf h
	- \mathbf h^H(\bm\Lambda_q + \mathbf C_q   + \bm\Lambda_c)\mathbf h - \psi_q \right\}  .
\end{IEEEeqnarray}
and
%
\begin{IEEEeqnarray}{Cl}
	\mathcal M_q
	= \left\{ p( \mathbf h; {\bm\theta}_q , {\bm\Theta}_q) \right\}.
\end{IEEEeqnarray}
These auxiliary manifolds $\mathcal M_q$s are parallel to each other when $\mathbf C_q$s are not diagonal since the points in different auxiliary manifolds never intersect.

The target manifold is still the manifold of independent complex Gaussian distributions
$\mathcal M_0
	= \left\{ p( \mathbf h; {\bm\theta}_0 , {\bm\Theta}_0)\right\}$. 
To be consistent with the auxiliary points, the natural parameter of the target point is defined as
\begin{IEEEeqnarray}{Cl}
	\subnumberinglabel{eq:tTheta0hat}
	{\bm\theta}_0 &= {\bm\lambda}_0 \\
	{\bm\Theta}_0 &= -( {\bm\Lambda}_0 +\bm\Lambda_c) 
\end{IEEEeqnarray} 
where ${\bm\Lambda}_0$ is diagonal.
The corresponding PDF $ p( \mathbf h; {\bm\theta}_0 , {\bm\Theta}_0)$ is
\begin{IEEEeqnarray}{Cl}
	\label{eq:tarpoint}
	 p( \mathbf h; {\bm\theta}_0 , {\bm\Theta}_0) 
	= \exp\left\{ \mathbf h^H {\bm\lambda}_0
	+ {\bm\lambda}_0^H\mathbf h
	- \mathbf h^H {(\bm\Lambda}_0 +\bm\Lambda_c)\mathbf h 
	- \psi_0  \right\}  
\end{IEEEeqnarray}
The target manifold is also parallel to all the auxiliary manifolds.

\begin{figure}
	\centering
	\includegraphics[width=0.6\linewidth]{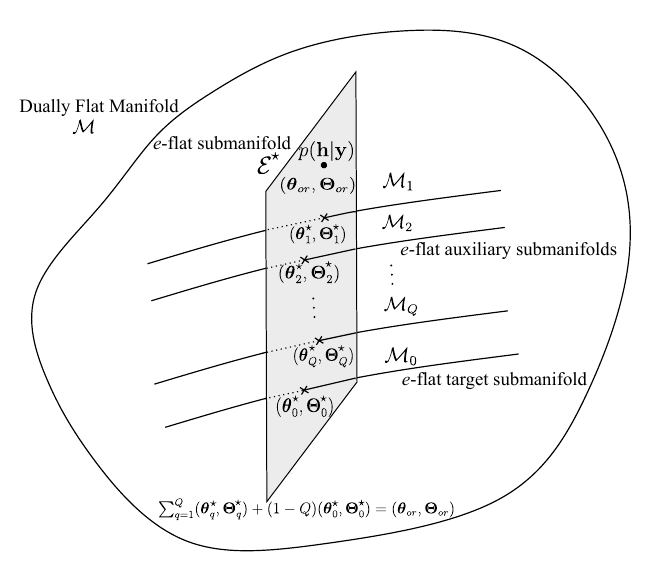}
	\caption{$e$-condition}
	\label{fig:econd}
\end{figure}

Now, we introduce the first condition of the information geometry framework as
\begin{IEEEeqnarray}{Cl}
	\label{eq:econdnew}
	\sum_{q=1}^Q({\bm\theta}_q,{\bm\Theta}_q) + (1-Q)({\bm\theta}_0,{\bm\Theta}_0) = (\bm\theta_{or},\bm\Theta_{or})
\end{IEEEeqnarray}
It means the original point, the target point and the auxiliary points are on a hyperplane in the affine coordinate system.
This is called the $e$-condition since it is a linear condition in the coordinate system $\bm\theta, \bm\Theta$ as shown in Fig.~\ref{fig:econd}.
Because of the split, the $e$-condition always holds if
\begin{IEEEeqnarray}{Cl} 
	\sum_{q=1}^Q(\bm\lambda_q,\bm\Lambda_q) + (1-Q)(\bm\lambda_0,\bm\Lambda_0) = 0.
	\label{eq:econ}
\end{IEEEeqnarray}
Thus, the above formula is also the $e$-condition. The $e$-condition makes sure the target point and auxiliary points are related to the original point, which is important to finally get an approximation of the $m$-projection point.

\subsection{The $m$-Condition and General Algorithm}
The $e$-condition only states the relation between the natural parameters of the original point, the auxiliary points, and the target point, but what we need is part of the expectation parameters of the original point. To obtain an approximation of the $m$-projection point, we need another condition, \textit{i.e}, the $m$-condition.

The $m$-condition is named because it is a linear condition in the dual affine coordinate system $\bm\mu, \mathbf{M}$, and is given by
\begin{IEEEeqnarray}{Cl}
\label{eq:m_condition}
	 ({\bm\mu}_q^{\star},\mathbf I\odot \mathbf{M}_q^{\star}) = ({\bm\mu}_0^{\star},\mathbf I\odot \mathbf{M}_0^{\star}  ) \quad \forall q\in \mathbb Z_Q^+ .
\end{IEEEeqnarray}  
From Section~\ref{eflat_and_mproj}, it means all the $m$-projection points of the auxiliary points are the same and equal to the target point. By combining the $m$-condition with the $e$-condition, a good approximation of the $m$-projection point on the target manifold of the original point can be obtained.

From \eqref{eq:dualmuM}, the expectation parameters ${\bm\mu}_q,{\mathbf M}_q$  of auxiliary points can be obtained as
\begin{IEEEeqnarray}{Cl}
	{\bm\mu}_q 
	& = ({\bm\Lambda}_q + \mathbf C_q +\bm\Lambda_c)^{-1} ({\bm\lambda}_q+\mathbf{b}_q ) 
	\IEEEyesnumber\IEEEyessubnumber*
	\label{eq:muqhat}\\
	{\mathbf M}_q
	& = {\bm\mu}_q{\bm\mu}_q^H
	+ ({\bm\Lambda}_q + \mathbf C_q +\bm\Lambda_c)^{-1}.
\end{IEEEeqnarray} 
Then, the expectation parameters of the $m$-projection points on the target manifold $\mathcal M_0$ satisfies
\begin{IEEEeqnarray}{Cl}
	\subnumberinglabel{eq:muMmproj0}
	{\bm\mu}_q ^0
	&= {\bm\mu}_q\\
	\mathbf I\odot {\mathbf M}_q ^0	
	&= \mathbf I\odot {\mathbf M}_q
\end{IEEEeqnarray}
as shown in Section~\ref{eflat_and_mproj}, and further we have the covariance of the $m$-projection ${\bm\Sigma}_q ^0
= \mathbf I \odot {\bm\Sigma}_q$.
The natural parameters of the $m$-projection points can be obtained as  
\begin{IEEEeqnarray}{Cl}
	\subnumberinglabel{eq:projtTheta}
	{\bm\theta}_q^0
	&=   - {\bm\Theta}_q^0
	({\bm\Lambda}_q + \mathbf C_q  +\bm\Lambda_c )^{-1} ({\bm\lambda}_q+\mathbf{b}_q )\\
	{\bm\Theta}_q^0
	& = - (\mathbf I\odot 
	({\bm\Lambda}_q + \mathbf C_q +\bm\Lambda_c)^{-1})^{-1}.
\end{IEEEeqnarray}

To make both the $e$-condition and the $m$-condition hold, the auxiliary points need to exchange beliefs.
Define ${\bm\lambda}_q^0$ and ${\bm\Lambda}_q^0$ to make ${\bm\theta}_q^0
= {\bm\lambda}_q^0$ and
${\bm\Theta}_q^0
= - ({\bm\Lambda}_q^0  +\bm\Lambda_c)$ hold.
The beliefs are defined as
\begin{IEEEeqnarray}{Cl}
	\subnumberinglabel{eq:xXi}
	\bm\xi_q 
	&={\bm\lambda}_q^0 
	- {\bm\lambda}_q \\
	\bm\Xi_q 
	&= {\bm\Lambda}_q^0
	- {\bm\Lambda}_q.
\end{IEEEeqnarray}
Then, the natural parameters of the $m$-projection points can be expressed as 
\begin{IEEEeqnarray}{Cl}
	\subnumberinglabel{eq:tThetaq0}
	{\bm\theta}_q^0
	&= {\bm\lambda}_q + \bm\xi_q  \\
	{\bm\Theta}_q^0
	&= - \left({\bm\Lambda}_q + \bm\Xi_q +\bm\Lambda_c \right).
\end{IEEEeqnarray}
By comparing them with the natural parameters of auxiliary point $ p( \mathbf h; {\bm\theta}_q , {\bm\Theta}_q)$, it can be observed that $\bm\xi_q$,  $\bm\Xi_q$ are approximations of $\mathbf b_q$,  $\mathbf C_q$ in ${\bm\theta}_q$ and ${\bm\Theta}_q$, respectively, where $\bm\Xi_q$ is also diagonal. 

After defining the beliefs, the iterative update of natural parameters of the target point and auxiliary points are constructed as
\begin{IEEEeqnarray}{Cl}
	\subnumberinglabel{eq:newlLambda0n}
	\bm\lambda_0^{t+1}&=\sum_q  \bm\xi_q^t\\
	\bm\Lambda_0^{t+1}&=\sum_q  \bm\Xi_q^t 
\end{IEEEeqnarray}
and
\begin{IEEEeqnarray}{Cl}
	\subnumberinglabel{eq:newlLambdaqn}
	\bm\lambda_q^{t+1}&=\sum_{q'\neq q} \bm\xi_{q'}^t
	= \bm\lambda_0^{t+1} - \bm\xi_q^t\\
	\bm\Lambda_q^{t+1}&=\sum_{q'\neq q} \bm\Xi_{q'}^t
	= \bm\Lambda_0^{t+1} - \bm\Xi_q^t.
\end{IEEEeqnarray}
From the above two equations, it is observed that the $e$-condition always holds.
An information geometry based algorithm is obtained by iteratively calculating \eqref{eq:projtTheta}, \eqref{eq:xXi}, \eqref{eq:newlLambda0n} and \eqref{eq:newlLambdaqn}.
When the algorithm converges, it is easy to obtain that 
\begin{IEEEeqnarray}{Cl}
	\subnumberinglabel{eq:tThetaq0}
	({\bm\theta}_q^0)^{\star}
	&= {\bm\lambda}_0^{\star}   \\
	({\bm\Theta}_q^0)^{\star}
	&= - \left({\bm\Lambda}_0^{\star} + \bm\Lambda_c \right) = {\bm\Theta}_0^{\star}
\end{IEEEeqnarray}
which means all the $m$-projection points are equal to the target point, and is equivalent to the $m$-condition. In conclusion, both the $e$-condition and the $m$-condition are satisfied when the algorithm converges. 

Finally, the output of this algorithm is the mean and covariance of the target point
\begin{IEEEeqnarray}{Cl}
	{\bm\mu}_0 
	&= -{\bm\Theta}_0^{-1}
	{\bm\theta}_0
	= {\bm\Lambda}_0^{-1}{\bm\lambda}_0
	\IEEEyesnumber\IEEEyessubnumber*\\
	{\bm\Sigma}_0
	&= -{\bm\Theta}_0^{-1}
	= {\bm\Lambda}_0^{-1}.
\end{IEEEeqnarray}
It is regarded as the approximated mean and covariance of the marginal PDF corresponding to the original PDF. When the algorithm does not converge, damping can be introduced in the updating of beliefs to ensure the convergence of the algorithm without changing its equilibrium. 
Although the process of the $m$-projection in \eqref{eq:projtTheta} still involves the matrix inversion, \textit{i.e.}, $({\bm\Lambda}_q + \mathbf C_q +\bm\Lambda_c)^{-1}$, it can be implemented with low complexity by properly setting of $\mathbf C_q$ since ${\bm\Lambda}_q$ and $\bm\Lambda_c$ are diagonal matrices. For example, in the IGA algorithm proposed in \cite{yang2022channel}, the matrix $\mathbf{C}_q$ is a rank-$1$ matrix.

%

\subsection{Geometrical Explanation}
This iterative process and the stabilization point can be explained geometrically.
Define the $m$-flat manifold $\mathcal M^{\star}$ and the $e$-flat manifold $\mathcal E^{\star}$ as
\begin{IEEEeqnarray}{Cl}
	\mathcal M^{\star}
	&= \left\{p(\mathbf{h}; \bm\theta,\bm\Theta)|(\bm\mu, \mathbf I\odot \mathbf M)=
	(\bm\mu_q^{\star},\mathbf I\odot \mathbf M_q^{\star}) = (\bm\mu_0^{\star},\mathbf I\odot \mathbf M_0^{\star}), \forall q \in \mathbb Z_Q^+\right\} 
	\IEEEyesnumber\IEEEyessubnumber*\\
	\mathcal E^{\star}
	&= \left\{p(\mathbf{h}; \bm\theta,\bm\Theta)|(\bm\theta,\bm\Theta)
	= \sum_{q=1}^Q c_q ({\bm\theta}_q^{\star},{\bm\Theta}_q^{\star}) + (1-\sum_{q=1}^Q  c_q) ({\bm\theta}_0^{\star},{\bm\Theta}_0^{\star})   \right\}
\end{IEEEeqnarray}
where $c_q$ are the positive coefficients. 
The geometric interpretation of the $m$-condition is given by Fig.~\ref{fig:mcond}.
The $e$-flat auxiliary and target manifolds are perpendicular to the $m$-flat manifold $\mathcal M^{\star}$  under the metric induced by the KL divergence.

\begin{figure}
	\centering
	\includegraphics[width=0.6\linewidth]{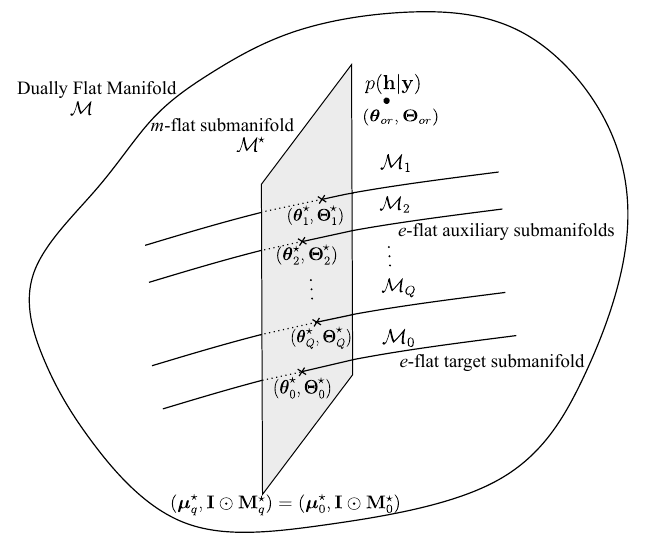}
	\caption{$m$-condition}
	\label{fig:mcond}
\end{figure}

From the $e$-condition and the $m$-condition shown in Figs.~\ref{fig:econd} and \ref{fig:mcond}, it can be seen that when the algorithm converges, all the auxiliary points, the target points and the original points are also on the same $e$-flat manifold $\mathcal E^{\star}$. Meanwhile, all the auxiliary points and the target point are on the same $m$-flat manifold $\mathcal M^{\star}$.
The $m$-condition makes sure the $m$-projection points of the auxiliary points and the target point are the same. The $e$-condition relates the auxiliary points and the target point to the original point. By combining these two conditions, the original point is close to the manifold $\mathcal M^{\star}$.
Theorem 11.8 of \cite{amari2016information} proves that $\mathcal M^{\star}$ contains the original point when the factor graph is acyclic, \textit{i.e.}, which means the exact mean of the original distribution is obtained. However, this property is not guaranteed when the factor graph is cyclic, which is the common case for the channel estimation problem.
Fortunately, for the channel estimation problem, it is usually able to prove the mean $\hat{\bm\mu}_0^{\star}$ is equal to the mean $\bm\mu_{or}$ when the algorithm converges as shown in \cite{yang2022channel}.


\section{IC-IGA for Massive MIMO Channel Estimation}
In this section, an IC-IGA  for massive MIMO channel estimation is proposed. First, a modified form equivalent to the MMSE estimation is given to define the auxiliary manifolds. Then, the framework of the information geometry approach is applied to derive this new IGA algorithm. Finally, the equilibrium and complexity analysis are provided.

\subsection{Definition of Auxiliary Manifolds with Modified MMSE Form}

In this subsection, we provide a new way of splitting the natural parameters based on a modified MMSE form. Then, the corresponding auxiliary manifolds are constructed. 


To derive a new low-complexity IGA algorithm, we want to make each auxiliary manifold focus on the computation of one element of $\mathbf{h}$.
According to the framework of information geometry methods, the $n$-th auxiliary PDF can be defined as
\begin{IEEEeqnarray}{Cl}
	p(\mathbf h; {\bm\theta}_n,{\bm\Theta}_n)
	= \exp\{\mathbf h^H({\bm\lambda}_{n} + \mathbf b_n) + ({\bm\lambda}_{n}^H + \mathbf b_n^H )\mathbf h - \mathbf h^H({\bm\Lambda}_{n} + \mathbf C_n)\mathbf h  - \psi_n\}
\end{IEEEeqnarray}	
where ${\bm\theta}_n={\bm\lambda}_{n} + \mathbf b_n$ , ${\bm\Theta}_n={\bm\Lambda}_{n} + \mathbf C_n$ and $\bm\Lambda_c$ is set to be a zero matrix.
 
Let $\mathbf K
= \sigma_z^{-2} \mathbf A^H \mathbf A$. 
To make the $n$-th auxiliary PDF only compute the information of the $n$-th element of $\mathbf{h}$, 
a natural idea is to set $\mathbf C_n$ and $\mathbf b_n$  as
\begin{IEEEeqnarray}{Cl}
	\mathbf P_{1n}\mathbf C_n\mathbf P_{1n}^H
	=\left(\begin{array}{cc}
		c_n & \frac{1}{2}\bar{\mathbf k}_n^H\\ 
		\frac{1}{2}\bar{\mathbf k}_n& \mathbf 0 
	\end{array}\right),
	\qquad
	\mathbf P_{1n}\mathbf b_n
	=\left(\begin{array}{c}
		\frac{1}{\sigma_z^2}\mathbf a_n^H\mathbf y\\ 
		\mathbf 0
	\end{array}\right)
\end{IEEEeqnarray}	
where $\mathbf a_n$ is the $n$-th column of $\mathbf A$,
 $\bar{\mathbf k}_n$ is the vector obtained by deleting the $n$-element $k_{nn}$ from the $n$-column of $\mathbf K$,  
$c_n =\sigma_z^{-2}\mathbf a_n^H\mathbf a_n + d_n^{-1}$, and $\mathbf P_{1n} \in \mathbb R^{N \times N}$ is the ordering matrix obtained by extracting $n$-th row of $\mathbf I_N$ and put it in the first row. The matrix $\mathbf P_{1n}$ has the property that $\mathbf P_{1n}\mathbf P_{1n}^H =\mathbf P_{1n}^H\mathbf P_{1n} = \mathbf I_N$.
By left-multiplying $\mathbf P_{1n}$ or right-multiplying $\mathbf P_{1n}^H$, one can extract the $n$-th row or column of a given matrix and put it in the first row or column.

However, there is a flaw in this setup. The items associated with elements other than $h_n$ in $\mathbf h$ are missing in the auxiliary function, and thus might not form a Gaussian-distributed PDF since ${\bm\Lambda}_{n} + \mathbf C_n$ are not always positive semidefinite. To overcome this issue, we present the following theorem.
\begin{theorem}
	\label{th: theorem MMSE equivalence}
	Let the matrices $\mathbf T $ and $\bm\Upsilon$ be defined as $\mathbf T = \sigma_z^{-2}\mathbf A^H\mathbf A - \mathbf I \odot (\sigma_z^{-2}\mathbf A^H\mathbf A)$ and 
$\bm\Upsilon = \left(\mathbf I \odot (\sigma_z^{-2}\mathbf A^H\mathbf A) + \mathbf D^{-1}\right)^{-1}$. The estimator
	\begin{IEEEeqnarray}{Cl}\label{eq:newMMSE}
		\hat{\mathbf h} = \left( \sigma_z^{-2}\mathbf A^H\mathbf A  + \mathbf D^{-1} + \mathbf T + \mathbf T \bm\Upsilon \mathbf T^H\right)^{-1} \left( \sigma_z^{-2} \mathbf A^H \mathbf y + \mathbf T \bm\Upsilon \sigma_z^{-2} \mathbf A^H \mathbf y  \right)
	\end{IEEEeqnarray}
	is equivalent to the MMSE estimator.
\end{theorem}
\begin{proof}
	The proof is provided in Appendix \ref{appendice: theorem MMSE equivalence}.
\end{proof}

Based on Theorem \ref{th: theorem MMSE equivalence}, we can use the following PDF 
\begin{IEEEeqnarray}{Cl}
	\label{eq:postEqual}
	 p( \mathbf h|\mathbf y)  &=  \exp\Big\{\mathbf h^H(\sigma_z^{-2} \mathbf A^H \mathbf y + \mathbf T \bm\Upsilon \sigma_z^{-2} \mathbf A^H \mathbf y  ) + (\sigma_z^{-2} \mathbf A^H \mathbf y + \mathbf T \bm\Upsilon \sigma_z^{-2} \mathbf A^H \mathbf y  )^H\mathbf h 
	\notag\\ 
	& \qquad \qquad  - \mathbf h^H\left(\sigma_z^{-2}\mathbf A^H\mathbf A  + \mathbf D^{-1} + \mathbf T + \mathbf T \bm\Upsilon \mathbf T^H\right)^{-1}\mathbf h - \tilde \psi \Big\} 
\end{IEEEeqnarray}	
to compute the original mean. 
With the modified MMSE form, we propose a new way of splitting natural parameter as
\begin{IEEEeqnarray}{Cl}
	\subnumberinglabel{eq:IGsplit}
	\bm\theta_{or} = \sigma_z^{-2}\mathbf{A}^H\mathbf y  + \mathbf T \bm\Upsilon \sigma_z^{-2} \mathbf A^H \mathbf y
	= \sum_{n=1}^N \mathbf{b}_n \\
	\bm\Theta_{or}  = -(\sigma_z^{-2}\mathbf{A}^H\mathbf{A} + \mathbf{D}^{-1}  + \mathbf T + \mathbf T \bm\Upsilon \mathbf T^H)
	= -(\sum_{n=1}^N \mathbf{C}_n)
\end{IEEEeqnarray}
where $\mathbf b_n$ and $\mathbf C_n$ are defined as
\begin{IEEEeqnarray}{Cl}
	\mathbf P_{1n}\mathbf C_n\mathbf P_{1n}^H
	=\left(\begin{array}{cc}
		c_n & \bar{\mathbf k}_n^H\\ 
		\bar{\mathbf k}_n& \frac{1}{c_n}\bar{\mathbf k}_n\bar{\mathbf k}_n^H 
	\end{array}\right),
	\qquad
	\mathbf P_{1n}\mathbf b_n
	=\left(\begin{array}{c}
		\frac{1}{\sigma_z^2}\mathbf a_n^H\mathbf y\\ 
		\frac{1}{\sigma_z^2} \frac{1}{c_n}  \bar{\mathbf k}_n \mathbf a_n^H\mathbf y
	\end{array}\right).
	\label{eq:turCbn} 
\end{IEEEeqnarray}	

The natural parameter of the $n$-th auxiliary PDF is defined as
\begin{IEEEeqnarray}{Cl}
\subnumberinglabel{eq:IGAuxiliary}
\bm\theta_n&=\bm\lambda_{-n} + \mathbf b_n \\
\bm\Theta_n&=-(\bm\Lambda_{-n} + \mathbf C_n)
\end{IEEEeqnarray}
where  $\bm\lambda_{-n}, \bm\Lambda_{-n}$ are given by 
\begin{IEEEeqnarray}{Cl}
	\bm\lambda_{-n}
	=(\lambda_1,\cdots ,\lambda_{n-1},0,\lambda_{n+1}, \cdots ,\lambda_N)^T
	\IEEEyesnumber\IEEEyessubnumber*\\
	\bm\Lambda_{-n}
	= \text{diag}(\Lambda_1,\cdots,\Lambda_{n-1},0,\Lambda_{n+1},\cdots,\Lambda_N).
\end{IEEEeqnarray}	
The subscript ${-n}$ instead of $n$ is used because we impose more constraints on $\bm\lambda_n, \bm\Lambda_n$, and the $n$-th element has been set to zero.
The auxiliary PDF $p(\mathbf h;\bm\theta_n,\bm\Theta_n)$ can be constructed as
\begin{IEEEeqnarray}{Cl}
	\label{eq:affTur}
	p(\mathbf h;\bm\theta_n,\bm\Theta_n)
	= \exp\{\mathbf h^H(\bm\lambda_{-n} + \mathbf b_n) + (\bm\lambda_{-n}^H + \mathbf b_n^H )\mathbf h - \mathbf h^H(\bm\Lambda_{-n} + \mathbf C_n)\mathbf h  - \psi_n\}.
\end{IEEEeqnarray}	
and the corresponding auxiliary manifolds are $\mathcal M_n
= \left\{p(\mathbf h;\bm\theta_n,\bm\Theta_n)\right\}, \forall n \in \mathbb Z_N^+$.


The natural parameter of the target point is defined as $\bm\theta_0=\bm\lambda$, 
$\bm\Theta_0=-\bm\Lambda$, 
where
\begin{IEEEeqnarray}{Cl}
	\bm\lambda &=( \lambda_1, \lambda_2,\cdots,  \lambda_N)^T
	\IEEEyesnumber\IEEEyessubnumber*\\
	\bm\Lambda &={\rm diag}( \Lambda_1, \Lambda_2,\cdots,  \Lambda_N).
\end{IEEEeqnarray}
The target PDF becomes
\begin{IEEEeqnarray}{Cl}
	p_0(\mathbf h;\bm\theta_0,\bm\Theta_0)
	= \exp\{\mathbf h^H\bm\lambda + \bm\lambda^H\mathbf h - \mathbf h^H\bm\Lambda\mathbf h  - \psi_0\}
\end{IEEEeqnarray}	
and the target manifold is still $\mathcal M_0
= \left\{p(\mathbf h;\bm\theta_0,\bm\Theta_0)\right\}$. It is easy to check the $e$-condition always holds due to 
\begin{IEEEeqnarray}{Cl} 
	\sum_{n=1}^N(\bm\lambda_{-n},\bm\Lambda_{-n}) + (1-N)(\bm\lambda,\bm\Lambda) = 0. 
\end{IEEEeqnarray}

\subsection{Derivation of IC-IGA}
After constructing the auxiliary manifolds, the estimation problem can be solved by applying the information geometry framework of section \ref{sec:IGAproce}.
For convenience, we define $\bm\lambda_{n}, \bm\Lambda_{n}$ as
\begin{IEEEeqnarray}{Cl}
	\bm\lambda_n
	=(\lambda_1,\cdots ,\lambda_{n-1},\lambda_{n+1}, \cdots ,\lambda_N)^T \in \mathbb C^{(N-1) \times 1}
	\IEEEyesnumber\IEEEyessubnumber*\\
	\bm\Lambda_n=\text{diag}(\Lambda_1,\cdots,\Lambda_{n-1},\Lambda_{n+1},\cdots,\Lambda_N)\in \mathbb C^{(N-1) \times (N-1)}.
\end{IEEEeqnarray}	

The following theorem gives the beliefs $\bm\xi_n$, $\bm\Xi_n$ corresponding to the $n$-th auxiliary point at the $t$-th iteration.

\begin{theorem}
	\label{th: theorem beliefs IC-IGA}
	Let the beliefs $\bm\xi_n$, $\bm\Xi_n$ be defined as $\bm\xi_n=\bm\lambda_{-n}^0 - \bm\lambda_{-n}$ and $\bm\Xi_n = \bm\Lambda_{-n}^0 -\bm\Lambda_{-n}$. Then, their elements are given by
	\begin{IEEEeqnarray}{Cl}
		[\bm\xi_n]_i  
		= \begin{cases}
			r_n \mu_n, \quad i=n\\
			0,   \quad \text{others}
		\end{cases} 
		\IEEEyesnumber\IEEEyessubnumber*\\
		{[\bm\Xi_n]}_{ii}  = \begin{cases}
			r_n, \quad i=n\\
			0,   \quad \text{others}
		\end{cases}
	\end{IEEEeqnarray}	
	where $\mu_n = \frac{1}{\sigma_z^2}c_n^{-1} \left(\mathbf a_n^H\mathbf y
	- \mathbf a_n^H\mathbf A \bm\Lambda^{-1} \bm\lambda
	+ \mathbf a_n^H\mathbf a_n \Lambda_n^{-1} \lambda_n\right)$, 
	$r_n = c_n(1 + e_n)^{-1}$, and
	$e_n = \frac{1}{c_n}\bar{\mathbf k}_n^H \bm\Lambda_n^{-1} \bar{\mathbf k}_n$.
\end{theorem}
\begin{proof}
	The proof is provided in Appendix \ref{appendice: theorem beliefs IC-IGA}.
\end{proof}

Before proceeding to the derivation of IC-IGA, we present more insights about Theorem \ref{th: theorem beliefs IC-IGA}.  
	Let $\bar{\mathbf h}_n$ and $s_n$  be defined as $\bar{\mathbf h}_n = (h_1,\cdots,h_{n-1},h_{n+1},\cdots,h_N)^T$ and 
\begin{IEEEeqnarray}{Cl}
	s_n 
	&= \frac{1}{\sigma_z^2}\mathbf a_n^H\mathbf y - \bar{\mathbf k}_n^H\bar{\mathbf h}_n
\end{IEEEeqnarray}	
we have that 
\begin{IEEEeqnarray}{Cl}
	  \frac{1}{Z_n} \exp\{\mathbf h^H\mathbf b_n + \mathbf b_n^H \mathbf h - \mathbf h^H \mathbf C_n\mathbf h\}  &= \frac{1}{Z_n} \exp \left\{h_n^* s_n + s_n^* h_n - h_n^*c_n h_n - \frac{s_n^* s_n}{c_n}\right\} 
\end{IEEEeqnarray}	
where $Z_n$ is the normalization constant. It means this part of  $p(\mathbf h;\bm\theta_n,\bm\Theta_n)$ can be viewed as a conditional PDF
of $h_n$ as $p(h_n|h_1,\cdots,h_{n-1},h_{n+1},\cdots,h_N,\mathbf y)$.
Furthermore, the remain part of $p(\mathbf h;\bm\theta_n,\bm\Theta_n)$ can be viewed as
\begin{IEEEeqnarray}{Cl}
	 \exp\{\mathbf h^H\bm\lambda_{-n}  + \bm\lambda_{-n}^H\mathbf h - \mathbf h^H \bm\Lambda_{-n}\mathbf h  - \psi\} = p(h_1|\mathbf y) \cdots p(h_{n-1}|\mathbf y) p(h_{n+1}|\mathbf y)\cdots p(h_N|\mathbf y).
\end{IEEEeqnarray}	
The corresponding PDF $p_n(\mathbf h;\bm\theta_n,\bm\Theta_n)$ can be rewritten as
\begin{IEEEeqnarray}{Cl}
	&\quad p(\mathbf h;\bm\theta_n,\bm\Theta_n) \notag\\
	&= p(h_n|h_1,\cdots,h_{n-1},h_{n+1},\cdots,h_N,\mathbf y)p(h_1|\mathbf y) \cdots p(h_{n-1}|\mathbf y) p(h_{n+1}|\mathbf y)\cdots p(h_N|\mathbf y).
	\IEEEeqnarraynumspace
\end{IEEEeqnarray}	
Thus, the marginal PDF of other elements will be the same as that provided by $\bm\lambda_{-n}, \bm\Lambda_{-n}$, and each auxiliary manifold only updates the marginal PDF of one element.
This explains why the $n$-th auxiliary manifold has a belief of $0$ for the other elements as shown in Theorem \ref{th: theorem beliefs IC-IGA}.

From  Theorem \ref{th: theorem beliefs IC-IGA}, the $n$-th auxiliary manifold only computes the mean $\mu_n$ and variance $r_n^{-1}$ of the $n$-th element of $\mathbf h$. The corresponding natural parameter is $\lambda_n = r_n\mu _n, \Lambda_n = r_n$.
In the computation, the mean and variance of the other elements are those of the other auxiliary manifolds in the previous iteration, which is consistent with the idea of interference cancellation (IC), and therefore this method is called IC-IGA.

After calculating the beliefs, the parameters $\bm\lambda$, $\bm\Lambda$ and $\bm\lambda_{-n}$, $\bm\Lambda_{-n}$ corresponding to the target and auxiliary points are updated based on \eqref{eq:newlLambda0n} and \eqref{eq:newlLambdaqn}. 
However, this update might cause the algorithm to diverge. 
By introducing damping, the convergence of the algorithm can be enhanced without changing its equilibrium. 
Let $0<\alpha\leq 1$ be the damping coefficient, the natural parameters $\lambda_n^{t+1}$ and $\Lambda_n^{t+1}$ corresponding to the target and auxiliary points in $\bm\lambda^{t+1}$, $\bm\Lambda^{t+1}$, $\bm\lambda_{-n}^{t+1}$ and $\bm\Lambda_{-n}^{t+1}$ are updated as
\begin{IEEEeqnarray}{Cl}
	\subnumberinglabel{eq:lLambdanTur}
	\lambda_n^{t+1}
	&= \alpha r_n^{t+1}\mu_n^{t+1}
	+ (1-\alpha)\lambda_n^t \\
	\Lambda_n^{t+1}
	&= \alpha r_n^{t+1}
	+ (1-\alpha)\Lambda_n^t,
\end{IEEEeqnarray}
where the computation of $\mu_n^t$ and $r_n^t$ can be rewritten as
\begin{IEEEeqnarray}{Cl}
	\subnumberinglabel{eq:munAndrnTur}
	\mu_n^{t+1} 
	&= \frac{1}{\sigma_z^2}c_n^{-1} \left(\mathbf a_n^H\mathbf y
	- \mathbf a_n^H\mathbf A (\bm\Lambda^t)^{-1}\bm\lambda^t
	+ \mathbf a_n^H\mathbf a_n (\Lambda_n^t)^{-1}\lambda_n^t\right) \label{eq:mu_n}\\
	r_n^{t+1} &= \frac{c_n}
	{1 + e_n^{t}}
\end{IEEEeqnarray}
where $e_n^t = \frac{1}{c_n}\bar{\mathbf k}_n^H (\bm\Lambda_n^t)^{-1} \bar{\mathbf k}_n$ and 
$c_n =\frac{1}{\sigma_z^2}\mathbf a_n^H\mathbf a_n + d_n^{-1}$. 
Besides, define the matrix $\mathbf L = (\mathbf A^H\mathbf A)\odot(\mathbf A^H\mathbf A)^*$ and vector $\mathbf v
= (\Lambda_1^{-1}, \Lambda_2^{-1},\cdots,\Lambda_N^{-1})^T$, 
then the computation of $e_n^t$ can be rewritten as
\begin{IEEEeqnarray}{Cl}\label{eq:enTur}
	e_n^t = \frac{1}{\sigma_z^4}c_n^{-1}(\mathbf e_n^T \mathbf L \mathbf v^t - \mathbf a_n^H\mathbf a_n (\Lambda_n^{-1})^t \mathbf a_n^H\mathbf a_n)
\end{IEEEeqnarray}	
where $\mathbf e_n = [\mathbf I_N]_{:,n}$ is the vector where only the $n$-th element is $1$ and the rest are all $0$s. 
The channel estimation IC-IGA is summarized as Algorithm \ref{algo:IC-IGA}. 

\begin{algorithm}
	\caption{IC-IGA for channel estimation}
	\label{algo:IC-IGA}
	\KwIn{The received signal $\mathbf y$, the \textit{a priori} covariance $\mathbf D$ of $\mathbf h$ and the maximal iteration number $T$.}
	\KwOut{The approximated posterior mean $\bm\mu^t$ and covariance $(\mathbf r^t)^{-1}$ of beam domain channel $\mathbf h$.}
	\textbf{Initialization} $t=0$, $\bm\mu^{t} = \mathbf 0$, 
	$\bm\lambda^{t}=\mathbf 0$, 
	$\bm\Lambda^{t}=\mathbf I$, $\mathbf{v}^t=\mathbf{1}$, 
	calculate $\mathbf A^H\mathbf y$, 
	$\mathbf A^H\mathbf A$, and
	$\mathbf L = (\mathbf A^H\mathbf A) \odot (\mathbf A^H\mathbf A)^*$.\\
	\While{$t\leq T$}{
		Calculate the $m$-projection of $N$ auxiliary points to the target manifold and the corresponding mean $\mu_n^{t+1}$ and covariance $(r_n^{t+1})^{-1}$ as
		\begin{IEEEeqnarray}{Cl}
			\mu_n^{t+1} 
			&= \frac{1}{\sigma_z^2}c_n^{-1} \left(\mathbf a_n^H\mathbf y
	- \mathbf a_n^H\mathbf A (\bm\Lambda^t)^{-1}\bm\lambda^t
	+ \mathbf a_n^H\mathbf a_n (\Lambda_n^t)^{-1}\lambda_n^t\right) \notag\\
			r_n^{t+1} &= \frac{c_n}
			{1 + e_n^{t}} \notag\\
			e_n^{t} &= \frac{1}{\sigma_z^4 c_n}(\mathbf e_n^T \mathbf L \mathbf v^t - \mathbf a_n^H\mathbf a_n  (\Lambda_n^t)^{-1} \mathbf a_n^H\mathbf a_n).\notag
		\end{IEEEeqnarray}
		Update the parameters of the target and auxiliary points as 
		\begin{IEEEeqnarray}{Cl}
	\subnumberinglabel{eq:lLambdanTur}
	\lambda_n^{t+1}
	&= \alpha r_n^{t+1}\mu_n^{t+1}
	+ (1-\alpha)\lambda_n^t \notag\\
	\Lambda_n^{t+1}
	&= \alpha r_n^{t+1}
	+ (1-\alpha)\Lambda_n^t. \notag 
\end{IEEEeqnarray} 
		$t=t+1$.
	}
	When the algorithm converges or $t>T$, output the posterior mean $\bm\mu^t$ and covariance $(\mathbf r^t)^{-1}$ of  $\mathbf h$.
\end{algorithm}

\subsection{Equilibrium and Complexity Analysis}

In this subsection, we present analysis for the equilibrium, \textit{i.e.}, fixed point or limit point, and the complexity of the IC-IGA.

From \eqref{eq:m_condition} and \eqref{eq:econ}, Algorithm \ref{algo:IC-IGA} satisfies the following conditions at the equilibrium
\begin{IEEEeqnarray}{Cl}
	\bm\mu^{\star} = \bm\mu_n^{\star},
	\quad
	\bm\Sigma^{\star} = \mathbf I \odot \bm\Sigma_n^{\star}  
	\label{eq:equimuSigma}\\
	\bm\lambda^{\star} = \frac{1}{N-1} \sum_{n=1}^N \bm\lambda_{-n}^{\star},
	\quad 
	\bm\Lambda^{\star} = \frac{1}{N-1} \sum_{n=1}^N \bm\Lambda_{-n}^{\star}
	\label{eq:equilLambda}
\end{IEEEeqnarray}
which leads to the following theorem.
\begin{theorem}
	\label{th: IC-IGA Equilibrium}
	At the equilibrium of IC-IGA, the mean $\bm\mu^{\star}$ of $p(\mathbf h;\bm\theta_0,\bm\Theta_0)$ is equal to the mean $\bm\mu_{\text{MMSE}}$ of the posterior distribution $p(\mathbf h|\mathbf y)$.
\end{theorem}
\begin{proof}
	The proof is provided in Appendix \ref{appendice: IC-IGA Equilibrium}.
\end{proof}

We now analyze the complexity of the IC-IGA algorithm in terms of time complexity (or computational complexity) and space complexity. 
In each iteration, the time complexity is mainly in the multiplication of matrix $\mathbf L$ and vector $\mathbf v$, where $\mathbf L \in \mathbb R ^{N \times N}$ and $\mathbf v \in \mathbb R^{N \times 1}$, so the complexity is $\mathcal O(N^2)$. Finally, the time complexity of this algorithm is $\mathcal O(TN^2)$, and $T$ is the number of iterations.
The free variables of this algorithm are $N$-dimensional vectors $\bm\mu$, $\mathbf r$, $\mathbf e$, $\bm\lambda$, $\text{diag}(\bm\Lambda)$, and the number of free variables is $5N$, so the space complexity of this algorithm is $\mathcal O(N)$.
When $T$ is small, the complexity of this algorithm is lower compared with the time complexity of $\mathcal O(N^3+N^2M)$ and the space complexity of $\mathcal O(N^3)$ in the MMSE algorithm.
In addition, compared with the time complexity of $\mathcal O(TMN)$ and the space complexity of $\mathcal O(MN)$ in the existing IGA algorithm \cite{yang2022channel}, the IC-IGA algorithm has a comparable time complexity and a lower space complexity when $M$ is comparable to $N$.
In practice, the channel dimension $N$ is often smaller than the received signal dimension $M$ due to the sparsity of the channel, and the IC-IGA algorithm has lower time complexity and space complexity compared to the IGA algorithm.

\section{IC-SIGA for Massive MIMO Channel Estimation with BSCM and ZC Sequences}
\label{sec:ICSIGA}
In this section, an IC-SIGA with lower complexity is proposed based on the IC-IGA by directly constructing the iterative update of the mean. To further reduce the complexity of IC-SIGA in practical systems, a massive MIMO with UPA is considered. 
By using BSCM and ZC sequences, a practical receive model is then established, and finally, the complexity of IC-SIGA is reduced by FFT and sparsity of the beam domain channel.   

\subsection{Derivation of IC-SIGA}
In this subsection, we provide the derivation of IC-SIGA, which can further reduce the complexity of channel estimation.

From the previous section, it is clear that the time complexity of the IC-IGA algorithm is mainly in the multiplication $\mathbf L\mathbf v$. 
This computation is related to the calculation of $r_n$ and $e_n$, which involves the computation of the corresponding variance of the $m$-projection of each auxiliary point.
In the following, we reconsider the computation of the mean corresponding to the $m$-projection of each auxiliary point. From \eqref{eq:mu_n} and $\bm\mu^{\star} =(\bm\Lambda^{\star})^{-1}\bm\lambda^{\star}$, we can obtain 
\begin{IEEEeqnarray}{Cl}
	\label{eq:mu_nvec}
	\bm\mu^{\star} 
	&= \frac{1}{\sigma_z^2} \left(\mathbf A^H\mathbf y
	- \mathbf A^H\mathbf A {\bm\mu}^{\star}
	+ (\mathbf I \odot \mathbf A^H\mathbf A) {\bm\mu}^{\star}\right)./\mathbf c
\end{IEEEeqnarray}
where $\mathbf c=(c_1,\cdots,c_N)^T$, 
$c_n = \frac{1}{\sigma_z^2}\mathbf a_n^H\mathbf a_n + d_n^{-1}$, and $. /$ denotes the element-by-element division of vectors.
Equation \eqref{eq:mu_nvec} indicates that $\bm\mu$ can be updated without $r_n$, $e_n$. 

However, equation \eqref{eq:mu_nvec} might also not converge. Thus, we also need to introduce the damping coefficient $0<\alpha\leq 1$.
Let $\bm\mu^{temp}$ be defined as
	\begin{IEEEeqnarray}{Cl}
		\label{eq:mu_ntemp}
		\bm\mu^{temp} 
		&= \frac{1}{\sigma_z^2} \left(\mathbf A^H\mathbf y
		- \mathbf A^H\mathbf A \bm\mu^t
		+ (\mathbf I \odot \mathbf A^H\mathbf A) \bm\mu^t\right)./\mathbf c
	\end{IEEEeqnarray}
where $\bm\mu^t$ is the mean at the $t$-th iteration.
Then, the mean of the target point $\bm\mu^{t+1}$ can be updated as
	\begin{IEEEeqnarray}{Cl}
		\bm\mu^{t+1} = \alpha\bm\mu^{temp} + (1-\alpha)\bm\mu^t.
	\end{IEEEeqnarray} 
When the algorithm converges, output the mean of the target point $\bm\mu^t$. 
The IC-SIGA for channel estimation is summarized as Algorithm \ref{algo:IC-SIGA}.

\begin{algorithm}
	\caption{IC-SIGA for channel estimation}
	\label{algo:IC-SIGA}
	\KwIn{The received signal $\mathbf y$, the \textit{a priori} covariance $\mathbf D$ of $\mathbf h$ and the maximal iteration number $T$.}
	\KwOut{The approximated posterior mean $\bm\mu^t$ of beam domain channel $\mathbf h$.}
	\textbf{Initialization} $t=0$, $\bm\mu^t = \mathbf 0$, 
	calculate $\mathbf A^H\mathbf y$, 
	$\mathbf I \odot \mathbf A^H\mathbf A$.\\
	\While{$t\leq T$}{
		Calculate the $m$-projection of $N$ auxiliary points to the target manifold and the corresponding mean $\bm\mu^{temp}$ as
		\begin{IEEEeqnarray}{Cl}
			\bm\mu^{temp} 
			&= \frac{1}{\sigma_z^2} \left(\mathbf A^H\mathbf y
			- \mathbf A^H\mathbf A \bm\mu^t
			+ (\mathbf I \odot \mathbf A^H\mathbf A) \bm\mu^t\right)./\mathbf c.
			\notag
		\end{IEEEeqnarray}
		Update the mean of the target point as
		\begin{IEEEeqnarray}{Cl}
			\bm\mu^{t+1} = \alpha\bm\mu^{temp} + (1-\alpha)\bm\mu^t. \notag
		\end{IEEEeqnarray}
		$t=t+1$.
	}
	When the algorithm converges or $t>T$, output the posterior mean $\bm\mu^t$ of  $\mathbf h$.
\end{algorithm}

For the equilibrium, we can have the following theorem by using similar methods as that in IC-IGA.
\begin{theorem}
	\label{th: IC-SIGA Equilibrium}
	At the equilibrium of IC-SIGA, the mean $\bm\mu^{\star}$ is equal to the mean $\bm\mu_{\text{MMSE}}$ of the posterior distribution $p(\mathbf h|\mathbf y)$.
\end{theorem}
\begin{proof}
	The proof is provided in Appendix \ref{appendice: IC-SIGA Equilibrium}.
\end{proof}

\subsection{System Configuration and Channel Model of Massive MIMO}

In this subsection, we consider a 3D massive MIMO system equipped with UPA. The system configuration and the BSCM are introduced to show the properties of the channel.

Consider a massive MIMO-OFDM system working in time division duplexing (TDD) mode.  The antenna array at the BS is a UPA with $M_r = M_zM_x$ antennas, where $M_z$ and $M_x$ are the numbers of vertical and horizontal antennas. All users are equipped with a single antenna.
The carrier frequency is $f_c$, and the wavelength is $\lambda_c$. 
The vertical and horizontal antenna spacings $d_z$ and $d_x$ are both set to half wavelength.
The number of OFDM subcarriers is $N_c$, of which $M_p$ training subcarriers are used to transmit the uplink pilot signal. 
Let the set of training subcarrier indexes be defined as $\mathcal L = \{\ell_0,\ell_1,\cdots,\ell_{M_p-1}\}$. 
Let $M_g$ and $T_s$ be the length of the cyclic prefix (CP)  and sampling interval, respectively. The subcarrier interval is $\Delta f = \frac{1}{N_c T_s}$ and the transmission bandwidth is $B = M_p \Delta f$.

The directional cosines for the $z$ and $x$ axis are defined as $u_r=\sin \theta_r$ and $v_r=\cos \theta_r\sin \phi_r$, where
$\theta_r$ and $\phi_r$ are polar and azimuthal angles
of arrival (AOA) at the BS. 
The space steering vector at the BS side is given by \cite{lu20242d} 
\begin{IEEEeqnarray}{Cl}
	\mathbf v(u_r, v_r)
	= \mathbf{v}_z(u_r) 
	\otimes \mathbf{v}_x(v_r)
	\in \mathbb C^{M_r \times 1}
\end{IEEEeqnarray}
where
\begin{IEEEeqnarray}{Cl}
	\mathbf v_z(u)
	= [1 ~~ e^{-j 2\pi \tfrac{d_z}{\lambda_c}  u} ~~  \cdots ~~ e^{-j 2\pi \tfrac{(M_{z}-1)d_z}{\lambda_c}  u }]^T \\
	\mathbf v_x(v)
	= [1 ~~ e^{-j 2\pi \tfrac{d_x}{\lambda_c} v} ~~  \cdots ~~ e^{-j 2\pi \tfrac{(M_{x}-1)d_x}{\lambda_c} v}]^T.
\end{IEEEeqnarray}
Let $u_i$ and $v_j$ be sampled directional cosines, defined as
\begin{IEEEeqnarray}{Cl}
	u_i
	&= \frac{2(i-1)-N_z}{N_z},~~i\in \mathbb Z_{N_z}^+ 
 \\
	v_i
	&= \frac{2(j-1)-N_x}{N_x},~~j\in \mathbb Z_{N_x}^+ .  
\end{IEEEeqnarray}
Based on the space steering vector, the matrix of sampled space steering vectors are defined by
\begin{IEEEeqnarray}{Cl}
	\mathbf{V} = \mathbf{V}_z \otimes \mathbf{V}_x \in \mathbb{C}^{M_{r} \times N_{r}}
\end{IEEEeqnarray}
where
\begin{IEEEeqnarray}{Cl}
	\label{eq:V}
	\mathbf{V}_z&=[\mathbf{v}_z(u_1) ~~ \mathbf{v }_z(u_2) ~~ \cdots ~~ \mathbf{v}_z(u_{N_z})] 
	\IEEEyesnumber\IEEEyessubnumber*\\
	\mathbf{V}_x&=[\mathbf{v}_x(v_1) ~~ \mathbf{v }_x(v_2) ~~ \cdots ~~ \mathbf{v}_x(v_{N_x})].
\end{IEEEeqnarray}
The symbols $N_z=F_zM_z$ and $N_x=F_xM_x$ are the numbers of sampled horizontal and vertical cosines, where $F_z$ and $F_x$ are fine factors.
The frequency steering vector is defined similarly as
\begin{IEEEeqnarray}{Cl}
	\label{eq:utau}
	\mathbf u(\tau)
	= \left[1 ~~ e^{-j 2\pi \Delta f \tau } ~~  \cdots ~~ e^{-j 2\pi (M_p-1) \Delta f \tau }\right]^T.
\end{IEEEeqnarray}
The matrix of
sampled frequency steering vectors is then given by
\begin{IEEEeqnarray}{Cl}
	\label{eq:U}
	\mathbf U=[\mathbf u(\tau_1) ~~ \mathbf u(\tau_2) ~~ \cdots ~~ \mathbf u(\tau_{N_f})] \in \mathbb C^{M_p \times N_f}
\end{IEEEeqnarray}
where $N_p=F_pM_p$ is the number of sampled delays, $F_p$ is the fine factor, $N_f=\left\lceil \frac{N_pM_g}{N_c}\right\rceil$, and
$\tau_r$ are sampled delays, given by
\begin{IEEEeqnarray}{Cl} 
	\tau_r
	&= \frac{r-1}{N_p \Delta f},~~r\in \mathbb Z_{N_f}^+. 
	\label{eq:taur}
\end{IEEEeqnarray}

Finally, the space-frequency domain channel matrix of user $k$ can be expressed as \cite{yang2022channel,lu20242d} 
\begin{IEEEeqnarray}{Cl}
	\label{eq:BSCM}
	\mathbf G_k =\mathbf V \mathbf H_k \mathbf U^T
	\in \mathbb C^{M_r \times M_p}
\end{IEEEeqnarray}
where $\mathbf H_k \in \mathbb C^{N_r \times N_f}$ is the beam domain channel matrix. 
This model is called the BSCM. The beam-domain channel matrix $\mathbf H_k$ has improved sparsity than the traditional beam-domain stochastic channel model based on the discrete Fourier transform (DFT) matrices \cite{sun2015beam}. The elements of $\mathbf H_k$ are assumed to be independent and follow a complex Gaussian distribution with zero mean and different variances.
The beam domain channel power matrix is defined as
$\bm\Omega_k$, where $[\bm\Omega_k]_{ij}$ is the variance of $[\mathbf H_k]_{ij}$.
It is the statistical CSI and remains constant over a relatively longer period than the instantaneous CSI.

\subsection{Receive Model with BSCM and ZC Sequences}
In this subsection, we describe a specific receive model in practical massive MIMO systems, which shows that $\mathbf{A}$ multiplying vector can be realized by FFT. 

The object of channel estimation is to obtain the \textit{a posteriori} information of the space-frequency channel $\mathbf G_k$, which can be calculated from the beam domain channel $\mathbf H_k$ with deterministic matrices $\mathbf U$ and $\mathbf V$. 
Thus, we focus on the estimation of $\mathbf H_k$. 
We use the pilot signal sequence in \cite{3GPP} as
\begin{IEEEeqnarray}{Cl}
	\mathbf x_k=\tilde{\mathbf x}_q \odot \mathbf u(\tau_{(p -1)N_f +1})
	\label{eq:pilots}
\end{IEEEeqnarray}
where $\tilde{\mathbf x}_q \in \mathbb C^{M_p \times 1}$ is the Zadoff-Chu (ZC) sequence defined as
\begin{IEEEeqnarray}{Cl}
	[\tilde{\mathbf x}_q]_l=e^{-j \frac{\pi (q-1) l(l-1)}{N_l}},\ l=1,\cdots,M_p
\end{IEEEeqnarray}
where $N_l$ is the largest prime number satisfying $N_l < M_p$, 
$q = \lfloor(k-1)/P\rfloor +1$ and $p = \left( (k-1)\bmod P \right) +1$ denote the root coefficient and cyclic shift, 
$Q = \lceil K/P\rceil$ is the number of roots, 
 $P$ is the number of UEs of each root, and
$\tau_{(p -1)N_f+1}$ is the sampled delay defined in \eqref{eq:taur}.

Let $\mathbf X_k=\text{diag}(\mathbf x_k) \in \mathbb C^{M_p \times M_p}$   be the pilot matrix. 
The received pilot signal matrix $\mathbf Y_t \in \mathbb C^{M_r \times M_p}$ at the $t$-th OFDM symbol is expressed as 
\begin{IEEEeqnarray}{Cl}
	\label{eq:reModelsf}
	\mathbf Y_t = \sum\limits_{k=1}^K
	\mathbf G_{k,t}\mathbf X_k + \mathbf Z_t
\end{IEEEeqnarray}
where the noise matrix $\mathbf Z_t$ consists of \textit{i.i.d.} elements with zero mean and variance $\sigma_z^2$. 

For convenience, the subscript $t$  of the OFDM symbol is omitted hereafter. 
Substituting \eqref{eq:BSCM} and \eqref{eq:pilots} into \eqref{eq:reModelsf}, we have
\begin{IEEEeqnarray}{Cl}
	\mathbf Y 
	&= \sum\limits_{q=1}^Q
	\mathbf V \left( \sum\limits_{p=1}^P \mathbf H_{q,p} \mathbf U^T \text{diag}\left( \mathbf u(\tau_{(p -1)N_f +1}) \right) \right) \tilde{\mathbf X}_q + \mathbf Z
	\label{eq:reciqp}
\end{IEEEeqnarray}
where $\tilde{\mathbf X}_q = \text{diag}\left( \tilde{\mathbf x}_q \right)$.
Let $\mathbf F_{N_p}$ be an $N_p$-dimensional DFT matrix and define
the partial DFT matrix $\mathbf U_F= [\mathbf u(\tau_1),\cdots,\mathbf u(\tau_{N_p})] \in \mathbb C^{M_p \times N_p}$. 
We then have $\mathbf U = \mathbf U_F \mathbf I_{N_p,N_f}$ and $\mathbf U_F = \mathbf I_{M_p,N_p}\mathbf F_{N_p}$.

In \eqref{eq:reciqp}, $\mathbf U^T \text{diag}\left( \mathbf u(\tau_{(p -1)N_f +1}) \right)$  can be calculated as
\begin{IEEEeqnarray}{Cl}
	\mathbf U^T \text{diag}\left( \mathbf u(\tau_{(p -1)N_f +1}) \right)
	&= \mathbf I_{N_f,N_p}\bm\Pi_{N_p}^{(p-1)N_f}\mathbf U_F^T
	\label{eq:IPiUF}
\end{IEEEeqnarray}
where 
\begin{IEEEeqnarray}{Cl}
	\bm\Pi_N^n
	= \left(\begin{array}{cc}
		\mathbf 0 & \mathbf I_{N-n}\\
		\mathbf I_n & \mathbf 0
	\end{array}\right)
\end{IEEEeqnarray}
is the permutation matrix.
Then, \eqref{eq:reciqp} can be reexpressed as 
\begin{IEEEeqnarray}{Cl}
	\mathbf Y
	&= \sum\limits_{q=1}^Q
	\mathbf V ( \sum\limits_{p=1}^P \mathbf H_{q,p} \mathbf I_{N_f,N_p} \bm\Pi_{N_p}^{(p-1)N_f} ) \mathbf U_F^T \tilde{\mathbf X}_q + \mathbf Z.
\end{IEEEeqnarray}
When $P\leq \lfloor N_p/N_f \rfloor$, we can avoid mutual aliasing of UEs with the same root and define
\begin{IEEEeqnarray}{Cl}
	\tilde{\mathbf H}_q 
	&= \left[ \mathbf H_{q,1}, \cdots,\mathbf H_{q,P}, \mathbf 0_{N_r,N_p-PN_f} \right]
	\in \mathbb C^{N_r \times N_p}.
\end{IEEEeqnarray}

Finally, the uplink received signal model is given as
\begin{IEEEeqnarray}{Cl}\label{mUreModel}
	\mathbf Y= 
	\mathbf V \mathbf H \mathbf P + \mathbf Z
\end{IEEEeqnarray}
where
\begin{IEEEeqnarray}{Cl}
	\mathbf H 
	&= [\tilde{\mathbf H}_1 ~~ \tilde{\mathbf H}_2 ~~ \cdots ~~ \tilde{\mathbf H}_Q]  \\
	\mathbf P 
	&= [\tilde{\mathbf X}_1\mathbf U_F ~~ \tilde{\mathbf X}_2\mathbf U_F ~~ \cdots ~~ \tilde{\mathbf X}_Q\mathbf U_F]^T.
	\label{eq:Pmtx}
\end{IEEEeqnarray}
By vectorizing \eqref{mUreModel}, the received signal model in vector form is
\begin{IEEEeqnarray}{Cl}\label{reModel2}
	\mathbf y= \tilde{\mathbf A}\tilde{\mathbf h}+ \mathbf z
\end{IEEEeqnarray}
where $\mathbf y = \text{vec}(\mathbf Y) \in \mathbb C^{M \times 1}$, 
$\tilde{\mathbf A} = \mathbf P^T \otimes \mathbf V \in \mathbb C^{M \times \tilde N}$, 
$\tilde{\mathbf h} = \text{vec}(\mathbf H) \in \mathbb C^{\tilde N \times 1}$, 
$\mathbf z = \text{vec}(\mathbf Z)  \sim \mathcal{CN}(\mathbf 0,\sigma_z^2\mathbf I)$, 
$M= M_r M_p$, and $\tilde N = QN_p N_r$.
Since the beam domain channels are sparse, we can obtain a low dimensional $\mathbf h = \tilde{\mathbf E} \mathbf h  \in \mathbb C^{N \times 1}$ from $\tilde{\mathbf h}$ by removing the elements with zero variance, where $\tilde{\mathbf E}$ is the extraction matrix. Then, the general received signal model $\mathbf y= {\mathbf A}{\mathbf h}+ \mathbf z$ in \eqref{reModel3} can be obtained, where 
$\mathbf A \in \mathbb C^{M \times N}$ is the matrix obtained by removing corresponding columns of $\tilde{\mathbf A}$, and $\mathbf D$ can be obtained from $\boldsymbol{\Omega}_k$ of all users.  

\subsection{Complexity Analysis} 

The time complexity of the IC-SIGA algorithm is mainly in the multiplication $\mathbf A^H \mathbf b$ and $\mathbf A \mathbf s$, where  $\mathbf s$ and $\mathbf b$ are two arbitrary vectors.
In the following, we analyze the fast implementation method and complexity of the operation $\mathbf A^H \mathbf b$. 
The computation of $\mathbf A \mathbf s$ can be implemented similarly. 

The computation of $\mathbf A^H \mathbf b$ can be written as $\mathbf A^H \mathbf b = \tilde{\mathbf A}^H \tilde{\mathbf b}$.
From $\tilde{\mathbf A}=\mathbf P^T \otimes \mathbf V$, the vector $\tilde{\mathbf A}^H \tilde{\mathbf b}$ can be transformed into a matrix as $\tilde{\mathbf A}^H \tilde{\mathbf b}
		= \text{vec}(\mathbf V^H \mathbf B \mathbf P^H)$,
	where $\mathbf B \in \mathbb C^{M_r \times M_p}$, $\text{vec}(\mathbf B) = \tilde{\mathbf b}$.
	Substituting the expression \eqref{eq:Pmtx} for $\mathbf P$  yields
$\mathbf B \mathbf P^H
		= \mathbf B (\tilde{\mathbf X}_1^*\mathbf I_{M_p,N_p}\mathbf F_{N_p}^*,\cdots,\tilde{\mathbf X}_Q^*\mathbf I_{M_p,N_p}\mathbf F_{N_p}^*)$.
Define $\tilde{\mathbf B}_q = [\mathbf B \tilde{\mathbf X}_q^*, \mathbf 0_{M_p,N_p-M_p}] \in \mathbb C^{M_r \times N_p}$, 
	then $\mathbf B \tilde{\mathbf X}_q^*\mathbf I_{M_p,N_p}\mathbf F_{N_p}^*
		= \tilde{\mathbf B}_q \mathbf F_{N_p}^*$,
where $\tilde{\mathbf B}_q \mathbf F_{N_p}^*$ can be realized by FFT.
	Since $\mathbf X_q^*$ is a diagonal matrix, the complexity of $\mathbf B \mathbf X_q^*$ is negligible.
	Thus, the time complexity of $\mathbf B \mathbf P^H$ is $\mathcal O(QM_rN_p\log_2 N_p)$.

From \eqref{eq:V} we have $\mathbf V
		= (\mathbf I_{M_z,N_z} \otimes \mathbf I_{M_x,N_x})
		(\mathbf F_{N_z} \otimes \mathbf F_{N_x}).$ 
	Then, $\mathbf V^H(\mathbf B \mathbf P^H)$ can be realized by FFT with a time complexity of $\mathcal O(QN_pN_r\log_2 N_r)$.
	In summary, the time complexity of IC-SIGA is $\mathcal O\left(TQ(M_rN_p\log_2 N_p + N_pN_r\log_2 N_r)\right)$
	b
	where $T$ is the iteration number.

Then, we analyze the space complexity. In IC-SIGA, the free variables are $N$-dimensional vectors $\bm\mu^{temp}$, $\bm\mu^t$, and the number of free variables is $2N$. Therefore, the space complexity of IC-SIGA is $\mathcal O(N)$.

\section{Simulation Results}
\begin{table}[htbp]
	\centering
	\caption{Parameter Setting of the QuaDRiGa}
	\label{tab:SimParaCh3}
	\renewcommand\arraystretch{0.85}
	\begin{tabular}{cc}
		\hline
		Parameter & Value\\
		\hline
		Number of BS antenna $M_r = M_z \times M_x$ 
		& 128 $=$ 8 $\times$ 16\\
		UT number $K$ 	& 12, 24\\
		Center frequency $f_c$	& 4.8GHz\\
		Number of subcarriers $N_c$	& 2048\\
		Subcarrier spacing $\Delta f$	& 30kHz\\
		Number of training subcarriers $M_p$	& 120\\
		CP length $M_g$	& 144\\
		\hline
	\end{tabular}
\end{table}
\begin{figure}[htbp] 
	\centering
	\includegraphics[width=0.65\linewidth]{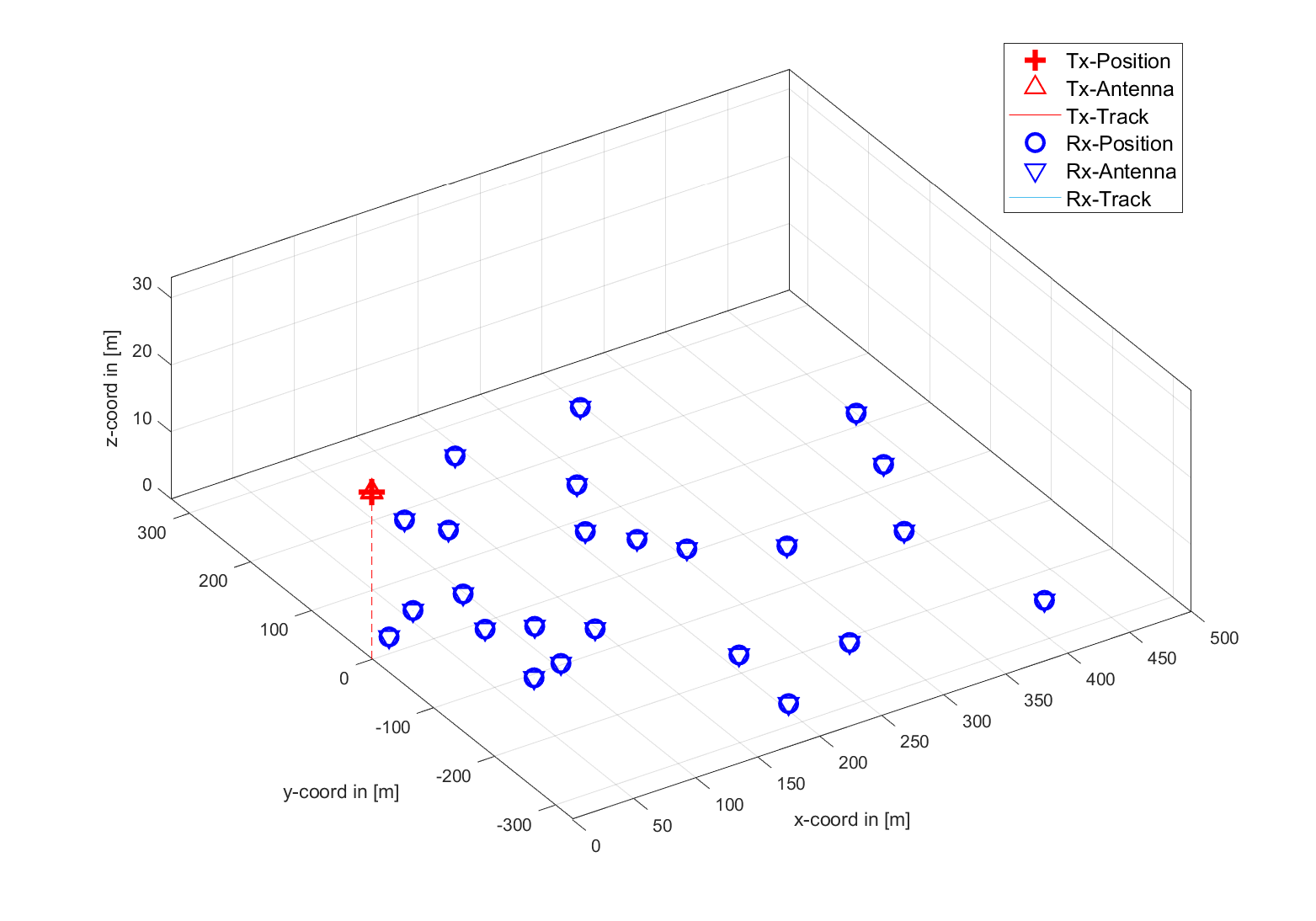}
	\caption{The layout of the massive MIMO-OFDM system.}
	\label{fig:UserPositions24U}
\end{figure}

In this section, we provide simulation results to show the
performance of the proposed IC-IGA and IC-SIGA for MIMO-OFDM channel estimation.
We use the widely used QuaDRiGa channel model. The simulation scenario is set to “3GPP$\_$38.901$\_$UMa$\_$NLOS”.
The main parameters for the simulations are summarized in Table \ref{tab:SimParaCh3}.
The layout of this massive MIMO-OFDM system is plotted in Fig.~\ref{fig:UserPositions24U}, where the location of the BS is at $(0, 0, 25)$, and the users are randomly generated in a $120^{\circ}$ sector with radius $r = 500$m around $(0, 0, 0)$ at $1.5$m height.
The channel is normalized as $\mathbb E\{\|\mathbf G_k\|_F^2\} = M_r M_p$.
Fine factors are set as $F_z=F_x=F_p=2$.
The SNR is set as SNR $=1/\sigma_z^2$.
We use the algorithm proposed in \cite{lu20242d} to obtain the beam domain channel power matrices $\bm\Omega_k,\forall k$. 
The normalized mean-squared error (NMSE) is used as the performance metric for the channel estimation, and is defined as
\begin{IEEEeqnarray}{Cl}
	\label{eq:NMSE}
	\text{NMSE} = \frac{1}{KN_{sam}}
	\sum\limits_{k=1}^{K}
	\sum\limits_{n=1}^{N_{sam}}
	\frac{\|\bar{\mathbf G}_{k,n} - \mathbf G_{k,n}\|_F^2}
	{\|\mathbf G_{k,n}\|_F^2}
\end{IEEEeqnarray}
where $N_{sam}$ is the number of the received pilot signals, $\mathbf G_{k,n}$ is the exact space-frequency domain channel matrix, and $\bar{\mathbf G}_{k,n}$ is the estimated space-frequency domain channel matrix.
We set $N_{sam} = 200$ in our simulations.

\begin{table}[htbp]
	\centering
	\caption{Complexities of algorithms}
	\label{tab:Ch3Comp}
	\renewcommand\arraystretch{0.75}
	\begin{tabular}{ccc}
		\hline
		Algorithm & Time Complexity & Space Complexity\\
		\hline
		IC-IGA & $\mathcal O(TN^2)$ & $\mathcal O(N)$\\
		IC-SIGA & $\mathcal O\left(TQ(M_rN_p\log_2 N_p + N_pN_r\log_2 N_r)\right)$ & $\mathcal O(N)$\\
		IGA & $\mathcal O(TMN)$ & $\mathcal O(MN)$\\
		MMSE  & $\mathcal O(N^3+N^2M)$ & $\mathcal O(N^3)$\\
		\hline
	\end{tabular}
\end{table}

\begin{figure}[htbp] 
	\centering
	\includegraphics[width=0.65\linewidth]{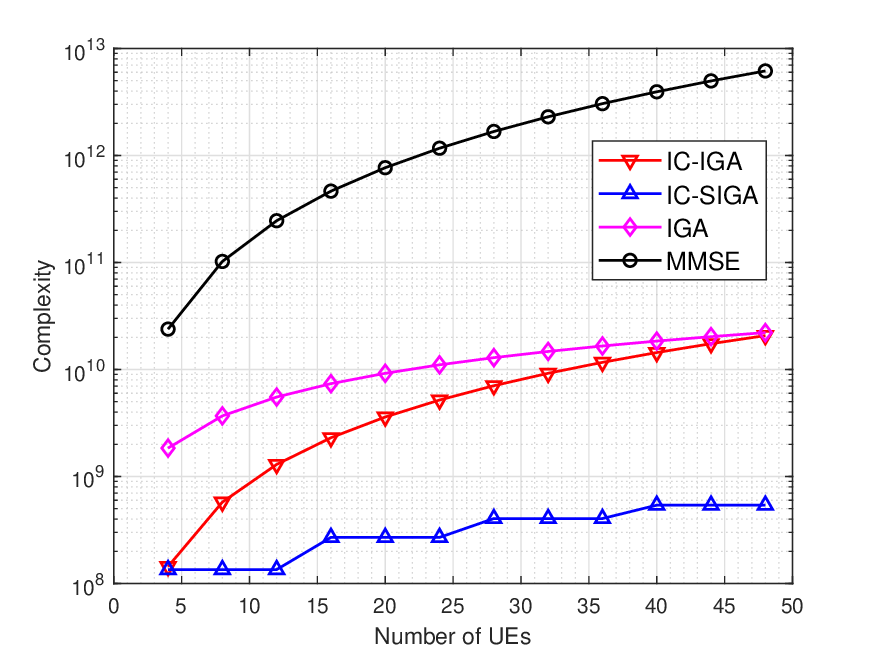}
	\caption{Time Complexity.}
	\label{fig:TComp}
\end{figure}

\begin{figure}[htbp] 
	\centering
	\includegraphics[width=0.65\linewidth]{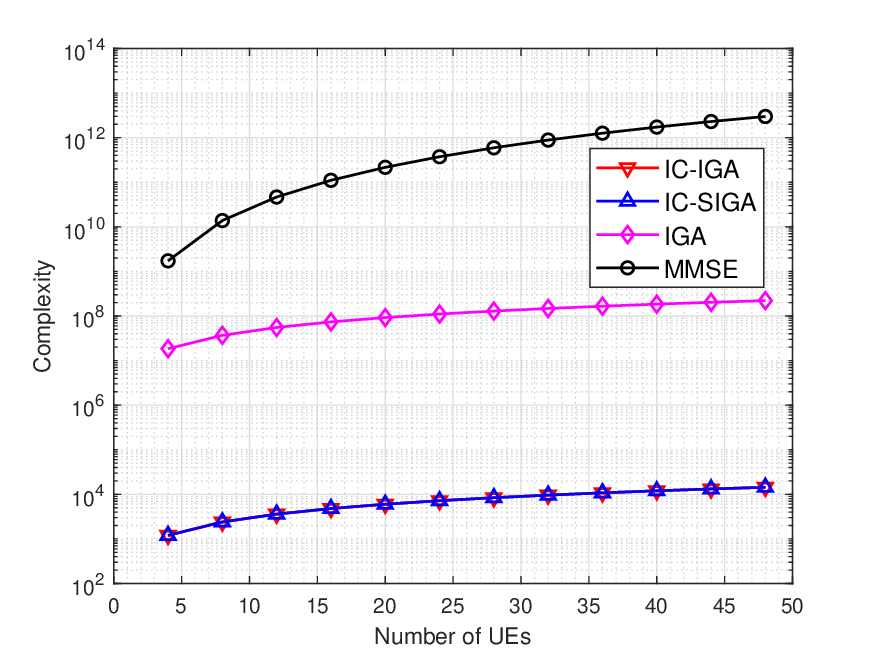}
	\caption{Space Complexity.}
	\label{fig:SComp}
\end{figure}

First, we present simulation results to show the complexity performance of different algorithms.
The complexity of the IC-IGA, IC-SIGA, IGA, and MMSE algorithms is summarized in Table \ref{tab:Ch3Comp}. 
The time complexity and space complexity of the above four algorithms are presented in Figs.~\ref{fig:TComp} and \ref{fig:SComp} for comparison, where the number of iterations is $T=100$, the number of pilot roots is $Q = \lceil K/P \rceil$, and simulation parameters are configured as in Table \ref{tab:SimParaCh3}.
From the figure, it can be found that the order of the time complexities of these algorithms is IC-SIGA$<$IC-IGA$<$IGA$<$MMSE, and the order of the space complexities of these algorithms is IC-SIGA$=$IC-IGA$<$IGA$<$MMSE. 
Meanwhile, the time complexity of IC-SIGA algorithm is the same for the same number of pilot roots.
The time complexity of IC-IGA is much less than the MMSE algorithm and also less than the IGA, whereas the space complexity of IC-IGA is much less than the IGA and MMSE algorithm.
The time complexity of IC-SIGA is much less than other algorithms, and the space complexity of IC-SIGA is much less than that of the IGA and the MMSE algorithm.

\begin{figure}[htbp] 
	\centering
	\includegraphics[width=0.65\linewidth]{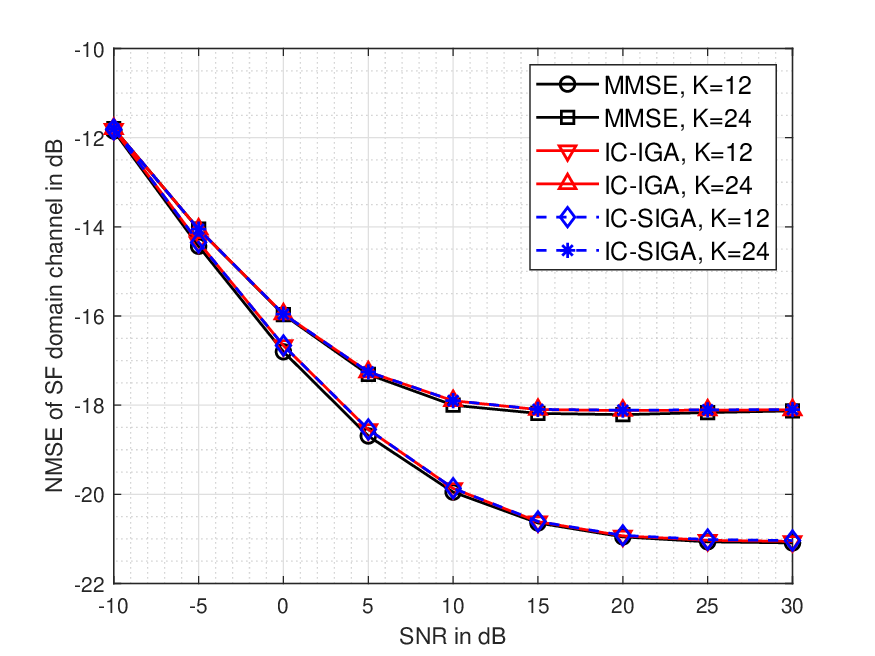}
	\caption{NMSE performance of IC-IGA and IC-SIGA compared with MMSE}
	\label{fig:TIGA_12U_24U}
\end{figure}

Fig.~\ref{fig:TIGA_12U_24U} shows the NMSE performance of IC-IGA and IC-SIGA
channel estimation compared with IGA and MMSE. 
The user number is set to be $K=12$ and $K=24$. 
The iteration numbers of IC-IGA and IC-SIGA are set as 100.
The damping coefficients of IC-IGA and IC-SIGA are $\alpha_{\text{IC-IGA}}=0.45$ and $\alpha_{\text{IC-SIGA}}=0.25$, respectively.
From the figure, it can be seen that both IC-IGA and IC-SIGA can obtain almost the same performance as MMSE for all SNR scenarios with two numbers of users. 
In addition, the performance using orthogonal pilots is more accurate, which is because the non-orthogonal pilots introduce interference from the pilots of users with other roots.

\begin{figure}[htbp] 
	\centering
	\includegraphics[width=0.65\linewidth]{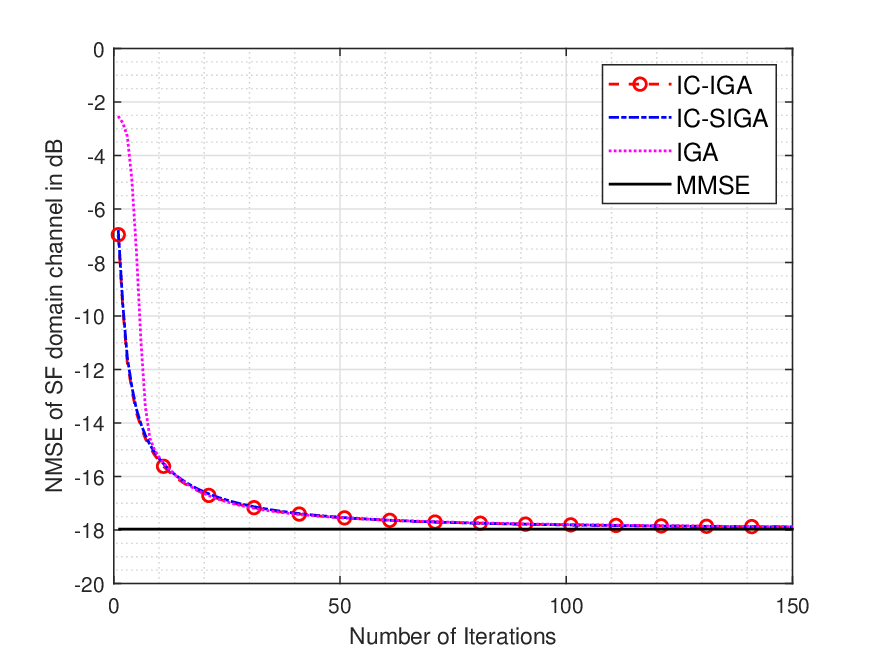}
	\caption{Convergence performance of IC-IGA and IC-SIGA at SNR=10dB}
	\label{fig:IGA_TIGA_Iter}
\end{figure}

Fig.~\ref{fig:IGA_TIGA_Iter} plots the convergence performance of the IC-IGA and IC-SIGA for $K=24$ at SNR=$10$dB. The results of the MMSE estimation and the IGA in \cite{yang2022channel} are given for comparison. The damping coefficients of IGA, IC-IGA and IC-SIGA are 
$\alpha_{\text{IGA}}=0.05$, $\alpha_{\text{IC-IGA}}=0.45$ and $\alpha_{\text{IC-SIGA}}=0.25$.
As shown in the figure, the NMSE performance of IC-IGA and IC-SIGA can converge close to that of the MMSE estimation, which is consistent with Theorems \ref{th: IC-IGA Equilibrium} and \ref{th: IC-SIGA Equilibrium}.
Furthermore, the NMSE of the IC-IGA decreases rapidly at the beginning, but the convergence speeds of the three algorithms, IC-IGA, IC-SIGA, and IGA, are nearly the same after about $10$ iterations.

\begin{figure}[htbp] 
	\centering
	\includegraphics[width=0.65\linewidth]{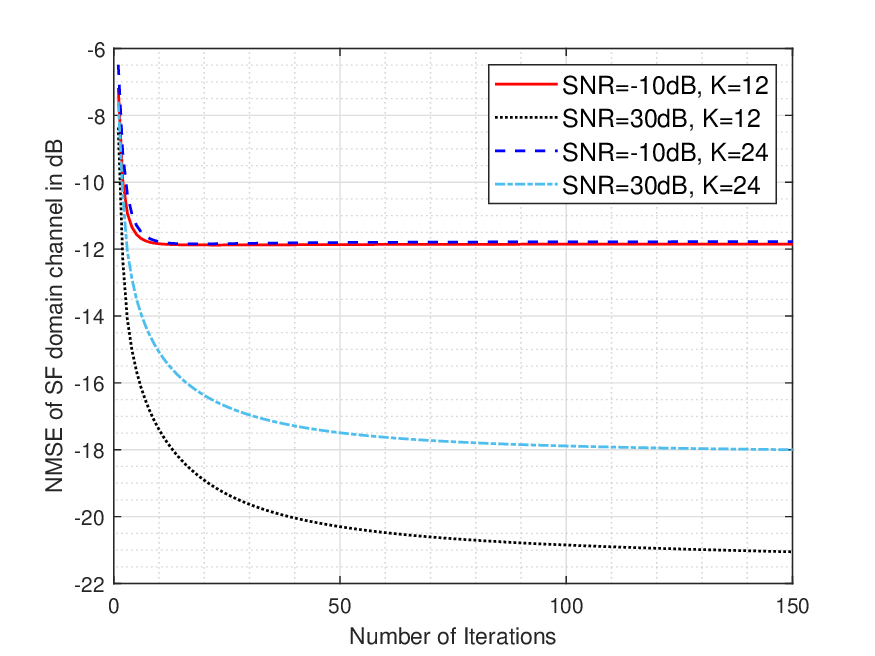}
	\caption{Convergence performance of IC-IGA}
	\label{fig:ICIGA_12U_24U_iter}
\end{figure}
\begin{figure}[htbp] 
	\centering
	\includegraphics[width=0.65\linewidth]{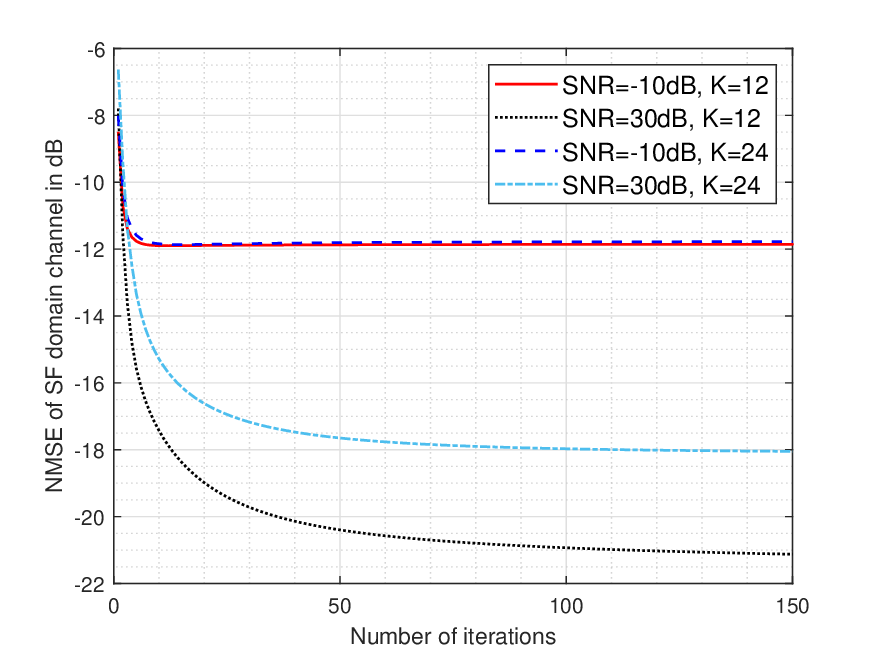}
	\caption{Convergence performance of IC-SIGA}
	\label{fig:ICSIGA_12U_24U_iter}
\end{figure}
Fig.~\ref{fig:ICIGA_12U_24U_iter} and Fig. \ref{fig:ICSIGA_12U_24U_iter} show the convergence performance curves of IC-IGA and IC-SIGA for the orthogonal pilots case ($K=12$) and non-orthogonal pilots ($K=24$) case at different SNRs. The SNRs under consideration are low SNR scenarios $\text{SNR}=-10\text{dB}$ and high SNR scenarios where $\text{SNR}=30\text{dB}$.
From the figure, it can be seen that IC-IGA and IC-SIGA can approach convergence within $10$ iterations in the low SNR scenario, and $60$ iterations in the high SNR scenario. 

\section{Conclusion}
In this paper, manifolds of complex Gaussian distributions are illustrated under information geometry theory, and a unified information geometry framework for channel estimation is described. 
To obtain an interference cancellation style algorithm, a modified MMSE form that has the same mean as the original MMSE estimator is constructed. 
Based on the unified framework and the modified form, the IC-IGA is then proposed for massive MIMO. The form of IC-IGA is simpler than IGA. Then, IC-SIGA is proposed to further reduce the complexity. 
The equilibria and complexities of the algorithms are analyzed.
Simulation results show that the proposed methods can obtain similar performance to the IGA algorithm with fewer iterations and lower complexity.

\appendices

\section{Proof of Theorem \ref{th: theorem MMSE equivalence}}
\label{appendice: theorem MMSE equivalence}
It is easy to verify that $\mathbf I + \mathbf T \bm\Upsilon$ is invertible. Simultaneously multiplying $\mathbf I + \mathbf T \bm\Upsilon$ inside and outside the matrix inversion in the MMSE estimator yields
\begin{IEEEeqnarray}{Cl}
	&\quad \left( \sigma_z^{-2}\mathbf A^H\mathbf A+ \mathbf D^{-1}\right)^{-1} \left( \sigma_z^{-2} \mathbf A^H \mathbf y\right) \notag\\
	&= \left( (\mathbf I + \mathbf T \bm\Upsilon)(\sigma_z^{-2}\mathbf A^H\mathbf A+ \mathbf D^{-1})\right)^{-1} (\mathbf I + \mathbf T \bm\Upsilon)\left( \sigma_z^{-2} \mathbf A^H \mathbf y\right) \notag\\
	&= \left( \sigma_z^{-2}\mathbf A^H\mathbf A  + \mathbf D^{-1} + \mathbf T \bm\Upsilon \sigma_z^{-2}\mathbf A^H\mathbf A + \mathbf T \bm\Upsilon \mathbf D^{-1} \right)^{-1}
	\left( \sigma_z^{-2} \mathbf A^H \mathbf y + \mathbf T \bm\Upsilon \sigma_z^{-2} \mathbf A^H \mathbf y  \right).
\end{IEEEeqnarray}
By using $\sigma_z^{-2}\mathbf A^H\mathbf A=\mathbf T^H + \mathbf I \odot (\sigma_z^{-2}\mathbf A^H\mathbf A)$, it follows that
\begin{IEEEeqnarray}{Cl}
&\quad \left( \sigma_z^{-2}\mathbf A^H\mathbf A+ \mathbf D^{-1}\right)^{-1} \left( \sigma_z^{-2} \mathbf A^H \mathbf y\right) \notag\\
	&= \left( \sigma_z^{-2}\mathbf A^H\mathbf A  + \mathbf D^{-1} + \mathbf T \bm\Upsilon \mathbf T^H + \mathbf T \bm\Upsilon (\mathbf I \odot (\sigma_z^{-2}\mathbf A^H\mathbf A))+ \mathbf T \bm\Upsilon \mathbf D^{-1} \right)^{-1}
	\left( \sigma_z^{-2} \mathbf A^H \mathbf y + \mathbf T \bm\Upsilon \sigma_z^{-2} \mathbf A^H \mathbf y  \right) \notag\\
	&= \left( \sigma_z^{-2}\mathbf A^H\mathbf A  + \mathbf D^{-1} + \mathbf T \bm\Upsilon \mathbf T^H+ \mathbf T \bm\Upsilon \bm\Upsilon^{-1} \right)^{-1}
	\left( \sigma_z^{-2} \mathbf A^H \mathbf y + \mathbf T \bm\Upsilon \sigma_z^{-2} \mathbf A^H \mathbf y  \right) \notag\\
	&= \left( \sigma_z^{-2}\mathbf A^H\mathbf A  + \mathbf D^{-1} + \mathbf T + \mathbf T \bm\Upsilon \mathbf T^H\right)^{-1} \left( \sigma_z^{-2} \mathbf A^H \mathbf y + \mathbf T \bm\Upsilon \sigma_z^{-2} \mathbf A^H \mathbf y  \right)
\end{IEEEeqnarray}
where the second equality is due to the definition of $\bm\Upsilon$. The above result means \eqref{eq:newMMSE} is equivalent to the MMSE estimator.

\section{Proof of Theorem \ref{th: theorem beliefs IC-IGA}}
\label{appendice: theorem beliefs IC-IGA}

The subscript $t$ is omitted here for convenience.
After placing the $n$-th row and $n$-th column of $(\bm\Lambda_{-n} + \mathbf C_n)$ in the first row and first column by left-multiplying of $\mathbf P_{1n}$ and right-multiplying of $\mathbf P_{1n}^H$, its inverse matrix can be computed by applying block matrix inversion formula, \textit{i.e.},
\begin{IEEEeqnarray}{Cl}
	\left(\mathbf P_{1n}(\bm\Lambda_{-n} + \mathbf C_n)\mathbf P_{1n}^H \right)^{-1}
	= \left(\begin{array}{cc}
		c_n & \bar{\mathbf k}_n^H \\ 
		\bar{\mathbf k}_n & \frac{1}{c_n}\bar{\mathbf k}_n\bar{\mathbf k}_n^H + \bm\Lambda_n
	\end{array}\right) ^{-1} \notag\\
	= \left(\begin{array}{cc}
		r_n^{-1} & -\frac{1}{c_n}\bar{\mathbf k}_n^H \bm\Lambda_n^{-1} \\ 
		-\frac{1}{c_n}\bm\Lambda_n^{-1}\bar{\mathbf k}_n  & \bm\Lambda_n^{-1}
	\end{array}\right) 
\end{IEEEeqnarray}	
where 
\begin{IEEEeqnarray}{Cl}\label{eq:r_n}
	r_n &= c_n - \bar{\mathbf k}_n^H\left(\frac{1}{c_n}\bar{\mathbf k}_n \bar{\mathbf k}_n^H + \bm\Lambda_n\right)^{-1}\bar{\mathbf k}_n. 
\end{IEEEeqnarray}
By using the Sherman-Morrison formula for matrix inversion, we can obtain that	
\begin{IEEEeqnarray}{Cl}	
	r_n &= 
	c_n - \bar{\mathbf k}_n^H\left(\bm\Lambda_n^{-1} - \frac{\bm\Lambda_n^{-1} \frac{1}{c_n}\bar{\mathbf k}_n\bar{\mathbf k}_n^H \bm\Lambda_n^{-1}}
	{1 + \frac{1}{c_n}\bar{\mathbf k}_n^H \bm\Lambda_n^{-1} \bar{\mathbf k}_n} \right)\bar{\mathbf k}_n \notag\\
	&= c_n - \left( \bar{\mathbf k}_n^H\bm\Lambda_n^{-1}\bar{\mathbf k}_n - \frac{\frac{1}{c_n}\bar{\mathbf k}_n^H\bm\Lambda_n^{-1} \bar{\mathbf k}_n\bar{\mathbf k}_n^H \bm\Lambda_n^{-1}\bar{\mathbf k}_n}
	{1 + \frac{1}{c_n}\bar{\mathbf k}_n^H \bm\Lambda_n^{-1} \bar{\mathbf k}_n} \right) \notag\\
	&= c_n -\frac{\bar{\mathbf k}_n^H\bm\Lambda_n^{-1} \bar{\mathbf k}_n}
	{1 +\frac{1}{c_n}\bar{\mathbf k}_n^H \bm\Lambda_n^{-1} \bar{\mathbf k}_n} \notag\\
	&= \frac{c_n}
	{1 +\frac{1}{c_n}\bar{\mathbf k}_n^H \bm\Lambda_n^{-1} \bar{\mathbf k}_n} \notag\\
	&= \frac{c_n}
	{1 + e_n}
\end{IEEEeqnarray}	
where $e_n = \frac{1}{c_n}\bar{\mathbf k}_n^H \bm\Lambda_n^{-1} \bar{\mathbf k}_n$. It can be shown that $c_n = r_n(1+e_n)$.

Further, the natural parameters of the $m$-projection of the auxiliary point to the target manifold are
\begin{IEEEeqnarray}{Cl}
	\bm\Theta_n^0 &= -(\bm\Sigma_n^0)^{-1} 
	= -(\mathbf I \odot \bm\Sigma_n )^{-1} \notag\\
	&= -(\mathbf I \odot(\bm\Lambda_{-n} + \mathbf C_n)^{-1})^{-1} \notag\\
	&= - \mathbf P_{1n}^H \left(\begin{array}{cc}
		r_n & \mathbf 0^H \\ 
		\mathbf 0 & \bm\Lambda_n
	\end{array}\right) \mathbf P_{1n}.
\end{IEEEeqnarray}	
Since
$\bm\Theta_n^0 = - \bm\Lambda_{-n}^0$, then the belief $\bm\Xi_n$ is
\begin{IEEEeqnarray}{Cl}
	\bm\Xi_n &= \bm\Lambda_{-n}^0 - \bm\Lambda_{-n} \notag\\
	&= \mathbf P_{1n}^H 
	\left(\begin{array}{cc}
		r_n & \mathbf 0^H \\ 
		\mathbf 0 & \bm\Lambda_n
	\end{array}\right) \mathbf P_{1n}
	- \mathbf P_{1n}^H 
	\left(\begin{array}{cc}
		0 & \mathbf 0^H \\ 
		\mathbf 0 & \bm\Lambda_n
	\end{array}\right) \mathbf P_{1n}
	\notag\\
	&= \mathbf P_{1n}^H 
	\left(\begin{array}{cc}
		r_n & \mathbf 0^H \\ 
		\mathbf 0 & \mathbf 0
	\end{array}\right) \mathbf P_{1n}.
\end{IEEEeqnarray}	
The mean $\bm\mu_n^0$ of $m$-projection is
\begin{IEEEeqnarray}{Cl}
	\bm\mu_n^0 &= \bm\mu_n
	= (\bm\Lambda_{-n} + \mathbf C_n)^{-1} (\bm\lambda_{-n} + \mathbf  b_n)
	\notag\\
	&= \mathbf P_{1n}^H \left(\begin{array}{cc}
		r_n^{-1} & -\frac{1}{c_n}\bar{\mathbf k}_n^H \bm\Lambda_n^{-1} \\ 
		-\frac{1}{c_n}\bm\Lambda_n^{-1}\bar{\mathbf k}_n  & \bm\Lambda_n^{-1}
	\end{array}\right)  
	\left(\begin{array}{c}
		\frac{1}{\sigma_z^2}\mathbf a_n^H\mathbf y\\ 
		\frac{1}{\sigma_z^2} \frac{1}{c_n}  \bar{\mathbf k}_n \mathbf a_n^H\mathbf y + \bm\lambda_n
	\end{array}\right).
\end{IEEEeqnarray}	 
From 
\begin{IEEEeqnarray}{Cl}
	&\quad -\frac{1}{c_n}\bm\Lambda_n^{-1}\bar{\mathbf k}_n \frac{1}{\sigma_z^2}\mathbf a_n^H\mathbf y
	+ \bm\Lambda_n^{-1} \left(\frac{1}{\sigma_z^2} \frac{1}{c_n}  \bar{\mathbf k}_n \mathbf a_n^H\mathbf y + \bm\lambda_n\right) 
	= \bm\Lambda_n^{-1}\bm\lambda_n
\end{IEEEeqnarray}
we have
\begin{IEEEeqnarray}{Cl}
	\bm\mu_n^0 &= \bm\mu_n 
	= \mathbf P_{1n}^H
	\left(\begin{array}{c}
		\mu_n\\ 
		\bm\Lambda_n^{-1}\bm\lambda_n
	\end{array}\right)
\end{IEEEeqnarray}		
where 
\begin{IEEEeqnarray}{Cl}
	\mu_n &= r_n^{-1}\frac{1}{\sigma_z^2}\mathbf a_n^H\mathbf y 
	  -\frac{1}{c_n}\bar{\mathbf k}_n^H \bm\Lambda_n^{-1} \cdot \frac{1}{\sigma_z^2} \frac{1}{c_n}  \bar{\mathbf k}_n \mathbf a_n^H\mathbf y 
	  -\frac{1}{c_n}\bar{\mathbf k}_n^H \bm\Lambda_n^{-1} \bm\lambda_n.
\end{IEEEeqnarray}	
The second and third terms can be further simplified as
\begin{IEEEeqnarray}{Cl}
	&\quad -\frac{1}{c_n}\bar{\mathbf k}_n^H \bm\Lambda_n^{-1} \cdot \frac{1}{\sigma_z^2} \frac{1}{c_n}  \bar{\mathbf k}_n \mathbf a_n^H\mathbf y \notag\\
	&= -\frac{1}{\sigma_z^2} \frac{1}{(c_n)^2}
	\bar{\mathbf k}_n^H \bm\Lambda_n^{-1} \bar{\mathbf k}_n \mathbf a_n^H\mathbf y \notag\\
	&\overset{(a)}{=}  -\frac{1}{\sigma_z^2} \frac{1}{c_n} e_n \mathbf a_n^H\mathbf y \notag\\
	&\overset{(b)}{=}  -\frac{1}{\sigma_z^2}r_n^{-1}(1+e_n)^{-1} e_n \mathbf a_n^H\mathbf y
\end{IEEEeqnarray}
and	
\begin{IEEEeqnarray}{Cl}
	&-\frac{1}{c_n}\bar{\mathbf k}_n^H \bm\Lambda_n^{-1} \bm\lambda_n
	\overset{(c)}{=}  -r_n^{-1}(1+e_n)^{-1}\bar{\mathbf k}_n^H \bm\Lambda_n^{-1} \bm\lambda_n
\end{IEEEeqnarray}	
where the equality $\overset{(a)}{=} $ is because $e_n = \frac{1}{c_n}\bar{\mathbf k}_n^H \bm\Lambda_n^{-1} \bar{\mathbf k}_n$, and $\overset{(b)}{=}, \overset{(c)}{=}$ is because $c_n = r_n(1+e_n)$.
By using this simplified results, we further obtain
\begin{IEEEeqnarray}{Cl}
	\mu_n &= r_n^{-1}\frac{1}{\sigma_z^2}\mathbf a_n^H\mathbf y 
	- r_n^{-1}(1+e_n)^{-1} e_n \frac{1}{\sigma_z^2}\mathbf a_n^H\mathbf y
	- r_n^{-1}(1+e_n)^{-1}\bar{\mathbf k}_n^H \bm\Lambda_n^{-1} \bm\lambda_n \notag\\
	&= r_n^{-1}(1+e_n)^{-1} \frac{1}{\sigma_z^2}\mathbf a_n^H\mathbf y
	- r_n^{-1}(1+e_n)^{-1}\bar{\mathbf k}_n^H \bm\Lambda_n^{-1} \bm\lambda_n \notag\\
	&= r_n^{-1}(1+e_n)^{-1} \left(\frac{1}{\sigma_z^2}\mathbf a_n^H\mathbf y
	- \bar{\mathbf k}_n^H \bm\Lambda_n^{-1} \bm\lambda_n\right) \notag\\
	&\overset{(d)}{=}  c_n^{-1} \left(\frac{1}{\sigma_z^2}\mathbf a_n^H\mathbf y
	- \bar{\mathbf k}_n^H \bm\Lambda_n^{-1} \bm\lambda_n\right).
\end{IEEEeqnarray}	
where the equality $\overset{(d)}{=}$ is because $c_n = r_n(1+e_n)$, $\mu_n$ and $r_n^{-1}$ are the mean and variance of the $n$-th element of $\mathbf h$ computed from the $n$-th auxiliary manifold, respectively.

Further, the natural parameter $\bm\theta_n^0$ of the $m$-projection is
\begin{IEEEeqnarray}{Cl}
	\bm\theta_n^0 &= -\bm\Theta_n^0 \bm\mu_n^0 
	= \mathbf P_{1n}^H 
	\left(\begin{array}{c}
		r_n \mu_n\\ 
		\bm\lambda_n
	\end{array}\right).
\end{IEEEeqnarray}	
According to $\bm\theta_n^0 = \bm\lambda_{-n}^0$, the belief $\bm\xi_n$ is
\begin{IEEEeqnarray}{Cl}
	\bm\xi_n 
	&= \bm\lambda_{-n}^0 - \bm\lambda_{-n} \notag\\
	&= \mathbf P_{1n}^H 
	\left(\begin{array}{c}
		r_n \mu_n\\ 
		\bm\lambda_n
	\end{array}\right)
	- \mathbf P_{1n}^H 
	\left(\begin{array}{c}
		0\\ 
		\bm\lambda_n
	\end{array}\right) \notag\\
	&= \mathbf P_{1n}^H 
	\left(\begin{array}{c}
		r_n \mu_n\\ 
		\mathbf 0
	\end{array}\right).
\end{IEEEeqnarray}

\section{Proof of Theorem \ref{th: IC-IGA Equilibrium}}
\label{appendice: IC-IGA Equilibrium}
From \eqref{eq:muqhat} and \eqref{eq:turCbn}, at the equilibrium, we have
\begin{IEEEeqnarray}{Cl}
	\sum_{n=1}^N \bm\lambda_{-n}^{\star}
	&= \sum_{n=1}^N(\bm\Lambda_{-n} + \mathbf C_n) \bm\mu_n^{\star} 
	- \sum_{n=1}^N\mathbf b_n \notag\\
	&= \left(\sum_{n=1}^N \bm\Lambda_{-n}^{\star} + \sum_{n=1}^N \mathbf C_n\right)\bm\mu^{\star} 
	- \sum_{n=1}^N\mathbf b_n \notag\\
	&= \left(\sum_{n=1}^N \bm\Lambda_{-n}^{\star}
	+ \sigma_z^{-2}\mathbf A^H\mathbf A  + \mathbf D^{-1} + \mathbf T + \mathbf T \bm\Upsilon \mathbf T^H\right)\bm\mu^{\star} 
	\notag\\
	&\quad - \left( \sigma_z^{-2} \mathbf A^H \mathbf y + \mathbf T \bm\Upsilon \sigma_z^{-2} \mathbf A^H \mathbf y  \right).
\end{IEEEeqnarray}
Substituting  $\sum_{n=1}^N \bm\Lambda_{-n}^{\star}=(N-1)\bm\Lambda^{\star}$ in \eqref{eq:equilLambda} and $\bm\mu^{\star}
= \left(\bm\Lambda^{\star}\right)^{-1} \bm\lambda^{\star}$ into the above equation, we have
\begin{IEEEeqnarray}{Cl}
	\sum_{n=1}^N \bm\lambda_{-n}^{\star}
	&= \left((N-1)\bm\Lambda^{\star}
	+ \sigma_z^{-2}\mathbf A^H\mathbf A  + \mathbf D^{-1} + \mathbf T + \mathbf T \bm\Upsilon \mathbf T^H\right)\bm\mu^{\star} 
	\notag\\
	&\quad - \left( \sigma_z^{-2} \mathbf A^H \mathbf y + \mathbf T \bm\Upsilon \sigma_z^{-2} \mathbf A^H \mathbf y  \right) \notag\\
	&= (N-1)\bm\lambda^{\star}
	+ \left(\sigma_z^{-2}\mathbf A^H\mathbf A  + \mathbf D^{-1} + \mathbf T + \mathbf T \bm\Upsilon \mathbf T^H\right)\bm\mu^{\star} 
	\notag\\
	&\quad - \left( \sigma_z^{-2} \mathbf A^H \mathbf y + \mathbf T \bm\Upsilon \sigma_z^{-2} \mathbf A^H \mathbf y  \right).
\end{IEEEeqnarray}
According to $\sum_{n=1}^N \bm\lambda_{-n}^{\star}=(N-1)\bm\lambda^{\star}$ in \eqref{eq:equilLambda}, it can be further obtained that
\begin{IEEEeqnarray}{Cl}
	0 = \left(\sigma_z^{-2}\mathbf A^H\mathbf A  + \mathbf D^{-1} + \mathbf T + \mathbf T \bm\Upsilon \mathbf T^H\right)\bm\mu^{\star} 
	- \left( \sigma_z^{-2} \mathbf A^H \mathbf y + \mathbf T \bm\Upsilon \sigma_z^{-2} \mathbf A^H \mathbf y  \right)
\end{IEEEeqnarray}
which means
\begin{IEEEeqnarray}{Cl}
	\bm\mu^{\star} 
	= \left( \sigma_z^{-2}\mathbf A^H\mathbf A  + \mathbf D^{-1} + \mathbf T + \mathbf T \bm\Upsilon \mathbf T^H\right)^{-1} 
	\left( \sigma_z^{-2} \mathbf A^H \mathbf y + \mathbf T \bm\Upsilon \sigma_z^{-2} \mathbf A^H \mathbf y  \right).
\end{IEEEeqnarray}
From Theorem \ref{th: theorem MMSE equivalence} it follows that this equation is equal to the mean $\bm\mu_{\text{MMSE}}$ of the MMSE estimation, \textit{i.e.}, \begin{IEEEeqnarray}{Cl}
	\bm\mu^{\star} 
	= \left( \sigma_z^{-2}\mathbf A^H\mathbf A  + \mathbf D^{-1}\right)^{-1} \sigma_z^{-2} \mathbf A^H \mathbf y.
\end{IEEEeqnarray}

\section{Proof of Theorem \ref{th: IC-SIGA Equilibrium}}
\label{appendice: IC-SIGA Equilibrium}

From \eqref{eq:mu_nvec}, at the Equilibrium, we have
\begin{IEEEeqnarray}{Cl}
	\bm\mu^{\star}
	&= \frac{1}{\sigma_z^2} \left(\mathbf A^H\mathbf y
	- \mathbf A^H\mathbf A \bm\mu^{\star}
	+ (\mathbf I \odot \mathbf A^H\mathbf A) \bm\mu^{\star}\right)./\mathbf c.
\end{IEEEeqnarray}
It can be reexpressed as
\begin{IEEEeqnarray}{Cl}
	\left(\text{diag}(\mathbf c) \cdot \mathbf I + \sigma_z^{-2} \mathbf A^H \mathbf A - \sigma_z^{-2} \left( \mathbf I \odot \mathbf A^H \mathbf A\right)\right)\bm\mu^{\star}
	&= \sigma_z^{-2} \mathbf A^H\mathbf y.
\end{IEEEeqnarray}
From $c_n = \frac{1}{\sigma_z^2} \mathbf a_n^H\mathbf a_n + d_n^{-1}$, it follows that $\text{diag}(\mathbf c) = \sigma_z^{-2} \left( \mathbf I \odot \mathbf A^H \mathbf A\right) + \mathbf D^{-1}$, then it follows that
\begin{IEEEeqnarray}{Cl}
	\left(\sigma_z^{-2} \mathbf A^H \mathbf A + \mathbf D^{-1} \right)\bm\mu^{\star}
	&= \sigma_z^{-2} \mathbf A^H\mathbf y
\end{IEEEeqnarray}
which means
\begin{IEEEeqnarray}{Cl}
	\bm\mu^{\star} 
	= \left( \sigma_z^{-2}\mathbf A^H\mathbf A  + \mathbf D^{-1}\right)^{-1} \sigma_z^{-2} \mathbf A^H \mathbf y.
\end{IEEEeqnarray}

\bibliographystyle{IEEEtran}
\bibliography{IEEEabrv,this_reference}

\end{document}